\newif\ifacm\acmfalse
\ifacm
\documentclass[acmsmall, nonacm]{acmart}
\else
\documentclass[11pt]{article}
\fi

\ifacm
\else
\usepackage{amsmath, amssymb, amsthm, amscd, amsfonts}
\usepackage[margin=1in]{geometry}
\usepackage{url}
\usepackage{xcolor}
\usepackage[colorlinks=true,linkcolor=blue,citecolor=blue,allcolors=blue]{hyperref}
\hypersetup{
 colorlinks=true,
 linkcolor=blue,
 filecolor=blue,
 urlcolor=blue,
 allcolors=blue
}

\fi

\usepackage{fontspec}
\usepackage{fontenc}
\usepackage{polyglossia}
\setdefaultlanguage{english}
\setotherlanguages{tamil,telugu}
\newcommand{\tam}[1]{\foreignlanguage{tamil}{#1}}
\newcommand{\tel}[1]{\foreignlanguage{telugu}{#1}}

\newfontfamily\tamilfont[Script=Tamil]{Lohit-Tamil}[Path=./]
\newfontfamily\tamilfontsf[Script=Tamil]{Lohit-Tamil}[Path=./]
\newfontfamily\telugufont[Script=Telugu]{Lohit-Telugu}[Path=./]
\newfontfamily\telugufontsf[Script=Telugu]{Lohit-Telugu}[Path=./]

\usepackage{graphicx}
\usepackage{float}
\usepackage{mathrsfs}

\usepackage{enumitem}
\usepackage{xspace}
\usepackage{bbm}
\usepackage{thmtools,thm-restate}
\usepackage{import}
\usepackage{multicol}

\usepackage{algorithm, algorithmicx, algcompatible}
\usepackage[noend]{algpseudocode}

\usepackage{tikz}

\newcommand{\BigTimes}{\mathop{\vcenter{\hbox{\huge$\times$}}}} 

\theoremstyle{plain} \numberwithin{equation}{section}
\newtheorem{theorem}{Theorem}[section]

\newtheorem{fact}{Fact}

\newtheorem{lemma}[theorem]{Lemma}

\theoremstyle{definition}
\newtheorem{definition}[theorem]{Definition}

\newtheorem{remark}[theorem]{Remark}

\declaretheorem[name=Claim,numberwithin=section]{claim}

\ifacm
\usepackage{natbib}
\else
\fi

\newcommand{\true}[0]{\mbox{\em true}} 
\newcommand{\false}[0]{\mbox{\em false}} 

\newcommand{\N}[0]{\mathbb{N}} 
\newcommand{\nullset}[0]{\emptyset}

\newcommand{\set}[1]{\left\{#1\right\}} 
\newcommand{\p}[1]{\left(#1\right)}

\newcommand{\bra}[1]{\left\langle#1\right\vert} 
\newcommand{\ket}[1]{\left\vert#1\right\rangle} 
\newcommand{\braket}[2]{\left\langle #1 \vert #2 \right\rangle}

\newcommand{\wh}[1]{\widehat{#1}} 
\newcommand{\ol}[1]{\overline{#1}} 
\newcommand{\ul}[1]{\underline{#1}}

\newcommand{\cmdif}[0]{\textbf{if}\xspace}
\newcommand{\cmdthen}[0]{\textbf{then}\xspace}
\newcommand{\cmdelse}[0]{\textbf{else}\xspace}

\newcommand{\cmdand}[0]{\textbf{and}\xspace}

\newcounter{int}
\newcommand{\tab}[1][1]{\setcounter{int}{0}\loop\hspace{\algorithmicindent}\addtocounter{int}{1}\ifnum\value{int}<#1\repeat}

\newcommand{\Acal}[0]{\mathcal{A}}
\newcommand{\Bcal}[0]{\mathcal{B}}
\newcommand{\Ccal}[0]{\mathcal{C}}

\newcommand{\Gcal}[0]{\mathcal{G}}
\newcommand{\Hcal}[0]{\mathcal{H}}

\newcommand{\Lcal}[0]{\mathcal{L}}
\newcommand{\Mcal}[0]{\mathcal{M}}

\newcommand{\Rcal}[0]{\mathcal{R}}
\newcommand{\Scal}[0]{\mathcal{S}}

\newcommand{\cprec}[0]{\prec}   
  
\newcommand{\ceq}[0]{\asymp}   
   
\newcommand{\ncprec}{\nprec}    

\newcommand{\Q}[0]{\textsc{q}}      
\newcommand{\quantum}[0]{\Q}        

\renewcommand{\H}[0]{\Hcal}         
\newcommand{\qstate}[0]{\rho}       

\DeclareMathOperator{\tr}{tr}       

\newcommand{\classical}[0]{\textsc{c}}  

\newcommand{\CStates}[0]{\Sigma}        
\newcommand{\XCStates}[0]{T}               
\newcommand{\CStateshat}[0]{\wh{\CStates}}  

\newcommand{\cstate}[0]{\sigma}            
\newcommand{\xcstate}[0]{\tau}             
\newcommand{\xcstatehat}[0]{{\wh{\xcstate}}} 
\newcommand{\xcstatet}[0]{{\tilde{\xcstate}}}   

\newcommand{\waitset}[0]{waitset} 
\newcommand{\resp}[0]{res}  

\newcommand{\Sys}[0]{{\Scal}} 
\newcommand{\fstate}[0]{\Psi}       
\newcommand{\FStates}{\boldsymbol{\Psi}}
\newcommand{\fstatehat}[0]{{\wh{\Psi}}}       

\newcommand{\fstatet}[0]{{\tilde{\Psi}}}       

\newcommand{\Events}{E}
\newcommand{\Procs}[0]{\Pi} 
\newcommand{\proc}[0]{\pi} 
\newcommand{\proct}{\tilde{\pi}}
\newcommand{\Chans}[0]{\Gamma} 
\newcommand{\chan}[0]{\gamma} 
\newcommand{\gchan}[2]{{\boldsymbol{\chan}[#1 \to #2]}} 
\newcommand{\Msgs}[0]{\Mcal} 
\newcommand{\msg}[0]{\mu} 
\newcommand{\Comps}[0]{\Ccal} 
\newcommand{\comp}[0]{c} 

\newcommand{\A}[0]{{\Acal}}     
\newcommand{\At}[0]{\tilde{\A}}  

\newcommand{\deltahat}[0]{\wh{\delta}}  

\newcommand{\Hist}{\mathscr{H}}

\newcommand{\X}[0]{X}           
\newcommand{\Xprime}[0]{{\X'}}     
\newcommand{\Xprimeprime}[0]{{\X''}}     
\newcommand{\Xhat}[0]{\wh{\X}}  
\newcommand{\Xt}[0]{{\tilde{\X}}}    

\newcommand{\Y}[0]{Y}           
\newcommand{\Yprimeprime}[0]{{\Y''}}     
\newcommand{\Yhat}[0]{\wh{\Y}}  
\newcommand{\Yt}[0]{{\tilde{\Y}}}    

\newcommand{\Zt}[0]{{\tilde{Z}}}    

\newcommand{\Wt}{\tilde{W}}

\newcommand{\U}[0]{U}      
\newcommand{\V}[0]{V}       
\newcommand{\Ut}[0]{{\tilde{\U}}}          
\newcommand{\Vt}[0]{{\tilde{\V}}}           

\newcommand{\E}[0]{E}  
\newcommand{\Ehat}[0]{\wh{\E}}  
\newcommand{\Et}[0]{\tilde{\E}}  

\newcommand{\GOps}[0]{\Gcal}    
\newcommand{\gop}[0]{G}         

\newcommand{\cop}[0]{\Ccal}         
\newcommand{\qop}[0]{\Lambda}    

\newcommand{\Res}[0]{\Rcal}     
\newcommand{\res}[0]{R}         

\newcommand{\Receive}[0]{\textsc{Receive}} 
\newcommand{\Broadcast}[0]{\textsc{Broadcast}} 

\newcommand{\send}{\mathsf{send}}
\newcommand{\receive}{\mathsf{receive}}

\newcommand{\invoke}{\mathsf{invoke}}
\newcommand{\respond}{\mathsf{respond}}

\newcommand{\Invoke}[0]{\textsc{Invoke}}
\newcommand{\Respond}[0]{\textsc{Respond}}

\newcommand{\Execute}[0]{\textsc{Execute}}
\newcommand{\NewMessage}[0]{\textsc{NewMessage}}

\newcommand{\cc}{\mathop{::}}

\newcommand{\len}[0]{\textit{len}\xspace} 
\newcommand{\contents}[0]{\textit{contents}\xspace} 
\newcommand{\lab}[0]{\textit{label}\xspace} 
\newcommand{\op}[0]{\textit{op}\xspace} 

\newcommand{\Append}[0]{\textsc{Append}}
\newcommand{\ProcessNewGlobalOp}[0]{\textsc{ProcessNewGlobalOp}}

\newif\ifnotes
\notestrue
\ifnotes
    \newcommand{\anote}[1]{\textcolor{red}{(Anand: #1)}}
    \newcommand{\snote}[1]{\textcolor{blue}{(Siddhartha: #1)}}
\else 
    \newcommand{\anote}[1]{}
    \newcommand{\snote}[1]{}
\fi
\title{Asynchronous Quantum Distributed Computing: Causality, Snapshots, and Global Operations}
\ifacm
\author{Siddhartha 
Visveswara Jayanti (\tel{సిద్ధార్థ విశ్వేశ్వర జయంతి})}
\affiliation{\institution{Dartmouth} \city{Hanover} \state{NH} \country{USA}}
\email{svj@dartmouth.edu}
\author{Anand Natarajan (\tam{ஆனந்த் நடராஜன்})}
\affiliation{\institution{MIT} \city{Cambridge} \state{MA}
   \country{USA}}
\email{anandn@mit.edu}
\else
\author{Siddhartha Visveswara Jayanti (\tel{సిద్ధార్థ విశ్వేశ్వర జయంతి}) \thanks{\texttt{svj@dartmouth.edu}} \\ \small{Dartmouth} \and Anand Natarajan (\tam{ஆனந்த் நடராஜன்}) \thanks{\texttt{anandn@mit.edu}}\\\small{MIT}}
\fi
\date{April 9, 2026}

\begin{document}

\ifacm
\begin{abstract}
We initiate the study of asynchronous quantum distributed systems, focusing on the case of implementing atomic \emph{quantum global operations} that can be decomposed into a collection of local operations on the components of the system. A simple example of such an operation is a \emph{quantum snapshot} in which the whole system is instantaneously measured. Based on the classical snapshot algorithm of Chandy and Lamport, we design a quantum distributed algorithm to implement such decomposable global operations, which we call the \emph{QGO Algorithm}. The analysis of our algorithm shows that arguments based on Lamport's computational causality remain valid in the quantum world, even though, due to entanglement, causality is not manifest from the standard description of the system in terms of a (global) quantum state. Our other contributions include a formal model of quantum distributed computing, and a formal specification for the desired behavior of a global operation, which may be of interest even in classical settings (such as in the setting of randomized algorithms).    
\end{abstract}
\maketitle
\else
\maketitle
\begin{abstract}
We initiate the study of asynchronous quantum distributed systems, focusing on the case of implementing atomic \emph{quantum global operations} that can be decomposed into a collection of local operations on the components of the system. A simple example of such an operation is a \emph{quantum snapshot} in which the whole system is instantaneously measured. Based on the classical snapshot algorithm of Chandy and Lamport, we design a quantum distributed algorithm to implement such decomposable global operations, which we call the \emph{QGO Algorithm}. The analysis of our algorithm shows that arguments based on Lamport's computational causality remain valid in the quantum world, even though, due to entanglement, causality is not manifest from the standard description of the system in terms of a (global) quantum state. Our other contributions include a formal model of quantum distributed computing, and a formal specification for the desired behavior of a global operation, which may be of interest even in classical settings (such as in the setting of randomized algorithms).    
\end{abstract}
\fi

\section{Introduction}
The analysis of asynchronous distributed computer systems relies on an understanding of \emph{locality} and \emph{causality}, creating a link between distributed systems and fundamental physics. This link means that distributed systems could behave in fundamentally different ways in worlds with different physical models of time and causality: for instance, to handle a world described by special relativity (where there is no global time and observers do not agree on the temporal order of events), basic concepts like that of linearizability require modification~\cite{GilbertGolabRelativisticDistributedSystems}, and it was only recently shown that linearizable systems in the ``Newtonian'' sense remain linearizable in a relativistic world~\cite{Jayanti2025}. 

Could quantum mechanics make a meaningful difference here as well? It is well known that quantum entanglement, as described in the standard ``Schr\"{o}dinger'' formalism of quantum mechanics in terms of wave functions and unitary operators, can cause the \emph{appearance} of instantaneous, nonlocal action. This point was famously raised by Einstein, Podolsky, and Rosen~\cite{einstein1935can}, who observed that, in the standard quantum formalism, if two systems $A$ and $B$ are in the ``EPR pair state''
\[ \ket{\psi}_{AB} = \frac{1}{\sqrt{2}} (\ket{0}_A \ket{0}_B + \ket{1}_A\ket{1}_B),\]
a measurement of system $A$ will cause \emph{both} systems to \emph{instantaneously} collapse to either the state $\ket{0}_A\ket{0}_B$ or $\ket{1}_A\ket{1}_B$, no matter how far system $A$ is from system $B$. Moreover, Bell's theorem~\cite{bell1964einstein} shows that this apparent nonlocality is in some sense ``real'': there are empirically realizable scenarios where entangled quantum systems, through local operations, can produce correlations in their classical measurement outcomes that are impossible to obtain in any classical local system. Quantum nonlocality can even be used to show that constant-depth quantum circuits can solve a computational problem that requires a logarithmic-depth circuit classically~\cite{bravyi2018quantum}, which can be viewed as a distributed computational advantage if one imagines the qubits in the circuit to be located on separate processors. 

Despite this apparent or real nonlocality in quantum mechanics, it is widely known to (or perhaps believed by) quantum information scientists that causality in quantum computational systems works essentially the same way as it does classically: in particular, causal influences spread locally. For instance, entanglement does not permit messages to be sent instantaneously, faster than the speed of light~\cite{eberhard1978bell,eberhard1989quantum}, and ``light-cone'' arguments to analyze the spread of causal information in quantum circuits are very common (e.g.~\cite{YunchaoThesis}). These are specific examples, but what about more complex causal scenarios?  In particular, asynchronous distributed systems introduce considerations not present in the previous, essentially synchronous, examples. Unlike  a quantum circuit in which the locations and times of operations are fixed, in a distributed system, events can be controlled by an adversarial scheduler that can act arbitrarily, and one wishes to prove that certain properties hold for \emph{every} possible execution consistent with the adversary's power. Moreover, a distributed system can receive input and produce output at many different times, unlike a quantum circuit to compute a function which receives its input all at once and produces its output all at once. Do these additional considerations make a material difference to the behavior of causality, and if not, how do we prove mathematically that causality continues to behave as it does classically? 

Apart from this basic scientific motivation, there are strong practical reasons to study quantum asynchronous distributed systems. Quantum computing hardware is advancing rapidly, and work is already underway to design and build networks of quantum devices, or a ``quantum internet''~\cite{wehner2018quantum}. Such networks are very likely to have a significant space-based component in order to minimize attenuation of the signals, which will likely be sent using photons~\cite{khatri2021spooky}. In fact, satellites implementing quantum computing hardware are already in orbit as technical demonstrations~\cite{ren2017ground}. 

Thus, in light of these theoretical and practical motivations, we initiate the study of asynchronous quantum distributed systems in this work. We focus our efforts on constructing and analyzing a quantum analog of the snapshot algorithm of Chandy and Lamport~\cite{ChandyLamport}, which Lamport has described~\cite{LamportWebsite} as ``a straightforward application of the basic ideas from~\cite{LamportTimeClocksOrderingOfEvents}.'' This algorithm is meant to implement the ideal functionality of taking a \emph{global snapshot} of a distributed system at a single point in time. In the quantum setting, because quantum states cannot be copied (the ``no-cloning theorem''), and measurement disturbs the state, it is not clear what the correct notion of a snapshot should be. We propose that the appropriate generalization is that of a \emph{decomposable global operation}: a global operation on the system that consists of a tensor product of separate operations on the components. This recovers the classical snapshot functionality when the operation is taken to be measurement (in the computational basis), but can also model many more classical and quantum behaviors---for instance, one may consider the \emph{global encryption} operation, in which one encrypts the state of all components of the system under a common key. Our main results are the following:
\begin{itemize}
    \item An extension of Lamport's definition of computational causality to quantum systems (Definition~\ref{def:causality}) and a proof that executions of a system that have the same causality relation (``equicausal'') behave equivalently: in particular, they result in the same final state (Theorem~\ref{thm:equiv}). 
    \item A formal specification of how an atomic global operation should act (Definition~\ref{def:atomic-spec}).
    \item An algorithm, which we call the QGO algorithm, and a proof that it implements the specification in the following sense: every execution of the QGO algorithm is \emph{equicausal} with an execution that is operationally indistinguishable from an execution of the specification (Theorem~\ref{thm:main}).
\end{itemize}
Our proof of Theorem~\ref{thm:main} is similar to Chandy and Lamport's analysis of their algorithm, although some extra steps are needed to handle the fact that our global operations modify the state. Nevertheless we view this result as showing that classical causality-based techniques are also applicable to quantum systems.

Our work opens up several interesting directions for future work. One direction is in fact purely classical: what is the strongest correctness guarantee we can get for global operations applied to \emph{classical randomized} systems? In particular, our result shows that any execution of the algorithm can be ``matched'' by some execution of the specification, but one could hope for a stronger statement relating the entire probability distribution of executions, for a fixed adversary, to the distribution over executions of the specification. It is possible that this will involve issues similar to those which arise in the study of linearizability in randomized algorithms, where it was shown that \emph{strong linearizability} is the more appropriate concept~\cite{StrongLinearizability}. Secondly, our work is restricted to quantum systems where an execution specifies the location and time at which each message between processes is sent or received. This corresponds to the typical setting of quantum circuits studied in quantum computing, but one could imagine scenarios in which \emph{superpositions} over different message histories are permitted. It has been speculated that causality may behave more exotically in these scenarios~\cite{brukner2014quantum}, and it would be interesting to explore this in the context of a distributed algorithm. Finally, what about quantum relativistic distributed systems? We already know that relativity is relevant for space-based systems, and it is likely that a future quantum internet will have a space-based component, so the quantum relativistic setting could be highly relevant in practice. 

\section{Other Related Work}
\label{sec:related-work}

The most directly related work was discussed in the introduction.
There have been several works on quantum distributed computing in synchronous settings, including works that find computational advantages over classical systems or algorithms for tasks that are not classically possible~\cite{denchev-pandurangan-2008, dufoulon-magniez-pandurangan-podc-2025, balliu-et-al-2025, A97}. 
However, to our knowledge, quantum asynchronous systems have not been studied.

\section{Quantum Physics and Distributed Systems}
\label{sec:quantum-primer}

We give a very brief and incomplete summary of the formalism and notation of quantum mechanics, in the slightly customized form used in this paper. This section is aimed at a reader with a basic knowledge of quantum computing, say from the first few weeks of an undergraduate class on the subject, and can safely be skimmed by quantum specialists.  For a fuller treatment of quantum theory, we recommend standard textbooks~\cite{nielsen2010quantum,kitaev2002classical}.

At any point in time, a quantum system is described by a complex \emph{Hilbert space} $\H$, which for our purposes can essentially always be thought of as a finite-dimensional complex vector space $\mathbb{C}^d$. The special case of $d =2$ is referred to as a \emph{qubit}. In Dirac's bra-ket notation---which is not needed for the body of the paper but which we will use for concrete examples in this section---(column) vectors in $\mathbb{C}^d$ will be denoted by \emph{kets}, e.g. $\ket{\psi}$, and dual vectors (i.e. row vectors) will be denoted by \emph{bras}, e.g. $\bra{\psi}$. 
This notation yields the pleasing expression $\braket{\alpha}{\beta}$ 
for the inner product between two vectors $\ket{\alpha}$ and $\ket{\beta}$. Matrices will be denoted by Latin or Greek letters: a useful way of writing a matrix in components is 
\[ M = \sum_{ij} m_{ij} \ket{i}\bra{j},\]
where $\ket{i}$ is the $i$th standard basis vector and $m_{ij}$ is the entry in the $i$th row and $j$th column of $M$.

The \emph{state} of a quantum system is a Hermitian positive semidefinite linear operator (i.e. a matrix) $\rho$ over $\H$ with trace at most one; such matrices are called (subnormalized) \emph{density matrices}, and the space of density matrices over $\H$ is denoted $D(\H)$. We use the notation $\Bcal(\H)$ to denote the space of bounded linear operators of $\H$, so $D(\H) \subset \Bcal(\H)$.

This formalism is different from the typical definition of quantum states as unit vectors $\ket{\psi}$ in $\H$ that is found in introductory treatments of quantum computing. Such states are called \emph{pure states}, and form a a subset of all possible quantum states: any such state can be captured by the associated rank-one density matrix $\ket{\psi}\bra{\psi}$.
Pure states include determinsitic classical states: for instance, for a qubit system, the qubit can be in the classical states $\ket{0}$ or $\ket{1}$. They also capture \emph{superpositions}, e.g. the state
\[ \ket{+} = \frac{1}{\sqrt{2}} (\ket{0} + \ket{1}). \]

However, the density matrix formalism allows to us to capture quantum states in their full generality: importantly, the state of a subsystem of a larger quantum system can always be described by a density matrix, but not always by a pure state (even if the global state of the entire system is pure). Moreover, a density matrix can describe a \emph{classical probability distribution} over quantum states. In fact our formalism is slightly more general than even the standard notion of density matrices, which requires that $\tr[\rho] = 1$. We allow for states with $\tr[\rho] < 1$ in order to keep track of probabilities of computational histories. In particular, suppose we imagine a system executing a sequence of operations, some of which are measurements with particular outcomes $r_1, r_2, \dots$, ending up in state $\rho$. Then, in our formalism,  $\tr[\rho]$ will be equal to the probability that the measurements performed in this sequence of operations actually produced the sequence $r_1, r_2, \dots$ of outcomes.

If two systems with Hilbert spaces $\H_A, \H_B$ are put together, the state space of the composite system is the \emph{tensor product} $\H_{AB} = \H_A \otimes \H_B$. The space of states $D(\H_{AB})$ over the tensor product space includes matrices that factor as tensor products, i.e. matrices of the form $\rho_A \otimes \rho_B$, but also many others. All states that are not convex combinations of tensor products are called \emph{entangled states}. A canonical example of an entangled state is the \emph{EPR state} on two qubits.
\[ \rho_{EPR} = \ket{EPR}\bra{EPR}, \quad  \ket{EPR} = \frac{1}{\sqrt{2}}(\ket{0}\otimes \ket{0} + \ket{1} \otimes \ket{1}). \]

Going the other way, given a state $\rho_{AB} \in D(\H_A \otimes \H_B)$, we can define the \emph{reduced state} of one of the systems---say, system $A$---by taking the partial trace:
\[ \rho_A = \tr_{B}[\rho_{AB}] = \sum_{i} (I \otimes \bra{i}) \rho_{AB} (I \otimes \ket{i}). \]
(Note that the subscript of $\tr$ is the system to be discarded!)  Here $\ket{i}$ is a standard basis vector in $\H_B$, and $i$ runs over the dimensions of $\H_B$. In our EPR example,
\[ \rho_{A} = \tr_B[ \rho_{EPR}] = \frac{1}{2}( \ket{0}\bra{0} + \ket{1}\bra{1}). \]

The possible operations on quantum systems are \emph{quantum operations}, which are linear maps acting on quantum states. These maps are required to satisfy certain positivity and normalization conditions, which we will specify below in Definition~\ref{def:quantum-operation}, but to start we will explain how these maps act at a basic level. To allow for systems to grow and shrink in size (e.g. by sending or receiving messages, or by initializing fresh qubits or discarding qubits), we allow a quantum operation to change the Hilbert space of the associated system. Thus, if we start with a system corresponding to Hilbert space $\H$, a quantum operation will be described by a linear map $\Lambda: \Bcal(\H) \to \Bcal(\H')$, where after the operation the system is now described by Hilbert space $\H'$: if the system started in state $\rho \in D(\H)$, it will end up in $\rho'= \Lambda(\rho) \in \Bcal(\H')$. A quantum operation can be applied to a subsystem of a composite system by taking its tensor product with the identity map: if we have a composite system $\H_{AB} = \H_A \otimes \H_B$ and an operation $\Lambda: \Bcal(\H_A) \to \Bcal(\H_A')$ on the $A$ system alone, then the operation acts on states of the composite system by
\[ \rho \mapsto (\Lambda \otimes \mathrm{id}) (\rho_{AB}) \in \Bcal(\H_A' \otimes \H_B). \]
Any operation on $AB$ that can be written in this way as a tensor product of an operation on $A$ and the identity on $B$ is referred to as being \emph{local to $A$}, or as acting \emph{trivially on $B$}, and these definitions can be straightforwardly generalized to systems with more components.

Often, treatments distinguish between \emph{measurements}, which act probabilistically and yield a classical outcome visible to classical observers, and other kinds of quantum operations, which act \emph{deterministically} on the state $\rho$ and do not yield a classical outcome. For us, it will be convenient to unify these objects into the single formalism of a quantum operation with  a (possibly empty) \emph{classical outcome set $R$}.

\begin{definition}[Quantum operation]
\label{def:quantum-operation}

A \emph{quantum operation with classical outcome set $R$} is a collection of linear maps $\Lambda^r: D(\H_{in}) \to D(\H_{out})$ for each $r \in R$, satisfying the following conditions: 
\begin{enumerate}
    \item Each $\Lambda^r$ is \emph{completely positive}, meaning that for any ancillary quantum system with Hilbert space $\H'$, and any state $\rho \in D(\H_{in} \otimes \H')$, $(\Lambda \otimes \mathrm{id})(\rho)$ is positive semidefinite.
    \item The linear map $\sum_{r \in R} \Lambda^r$ is \emph{trace-preserving},  meaning that for any $\rho \in D(\H)$, it holds that $\tr[\sum_{r \in R} \Lambda^r(\rho)] = 1$. 
\end{enumerate} 
Physically, when we apply this operation to a system in state $\rho$, a random outcome $r \in R$ is generated with probability $\Pr[r] = \tr[\Lambda^r(\rho)]$, and the system enters state $\Lambda^r(\rho)$.
\end{definition}

As an example of this formalism in action, let us consider the state $\rho_{EPR}$ as defined above, and consider applying a \emph{standard basis measurement} to the first qubit. The standard basis  measurement has outcome set $\{0,1\}$, and corresponds to the maps
\begin{align*}
    \Lambda^0(\rho) &= \ket{0}\bra{0} \cdot \rho \cdot \ket{0}\bra{0} \\
    \Lambda^1(\rho) &= \ket{1}\bra{1} \cdot \rho \cdot \ket{1}\bra{1}.
\end{align*}
Applying this operation to $\rho_{EPR}$ means we should apply $\Lambda^r \otimes \mathrm{id}$. 
\begin{align*}
    (\Lambda^0 \otimes \mathrm{id})(\rho_{EPR}) &= \frac{1}{2} (\ket{0}\bra{0})\otimes (\ket{0}\bra{0}), &\Pr[0] &= \frac{1}{2} \\
     (\Lambda^1 \otimes \mathrm{id})(\rho_{EPR}) &= \frac{1}{2} (\ket{1}\bra{1})\otimes (\ket{1}\bra{1}),  &\Pr[1] &= \frac{1}{2}.
\end{align*}
So we see that this operation yields a uniformly random outcome $r \in \{0,1\}$, and causes \emph{both} qubits to ``collapse'' to the classical state $\ket{r}$.  

\section{Model}
\label{sec:model}

A \emph{distributed system} $\Sys = (\Procs, \Chans)$ consists of a set of asynchronous \emph{processors}, $\Procs$, that communicate by passing \emph{messages}, from a universe $\Msgs$, along \emph{channels} $\Chans$.
The processors are \emph{asynchronous} in the sense that they may spend arbitrary and variable times between computational steps.
Message transmission is also \emph{asynchronous}, meaning that messages may take arbitrarily long to reach their intended recipient, but they eventually arrive.
For convenience, we assume that messages are delivered in order, i.e., that the channels are first-in-first-out (FIFO).
(If channels are not FIFO, senders can simply tag each message along a channel with sequence number and recipients can wait to process any message until all messages of lower sequence numbers along the channel are already received.)

\paragraph{Components and their States}
In general, a system is composed of several \emph{components}, $\Comps$, such as the processors, messages, and channels of a distributed system.
To capture systems of full generality, we model all components $\comp \in \Comps$ of our system to have both a \emph{classical part} $\classical_\comp$ and a \emph{quantum part} $\quantum_\comp$;
thus $\comp = \p{\classical_\comp, \quantum_\comp}$.
(For example, a message $\msg$ whose principal contents are quantum must still contain a classical tag in order for the message to be processed by the recipient without collapsing the quantum state.) 
The \emph{state space} of the classical part is a set of classical states denoted $\CStates_\comp$.
We generally use $\cstate_\comp \in \CStates_\comp$ to refer to the state of the classical part of the component. 
The \emph{state space} of the quantum part is a Hilbert space $\H_\comp$. 
We generally use $\qstate_\comp \in \H_\comp$ to refer to the (local) state of the quantum part of $\comp$.
It is noteworthy that while the classical state of the system is fully specified as the combination of the classical states of its components, due to entanglement, the quantum state of the entire system \emph{cannot} simply be expressed as the combination of quantum states of its components.
The quantum state of the system must be specified \emph{jointly}, as a density matrix of the entire quantum part of the system.
So, $\qstate_\comp$ captures the \emph{local} state of $\quantum_\comp$ and does not capture any entanglements that $\comp$ has with other parts of the system.

\paragraph{Channel States}
A \emph{channel} $\chan$ has a \emph{source} $\proc \in \Procs$ and a \emph{destination} $\proc' \in \Procs$.
We use the notation $\gchan{\proc}{\proc'}$ to refer to the channel whose source is $\proc$ and destination is $\proc'$.
We assume a single channel from every source to every destination, thus \emph{the set of all channels} is $\Chans = \set{\gchan{\proc}{\proc'} \mid \proc,\proc' \in \Procs}$.
Note that we assume a channel from each process to itself. 
At any point in time a channel $\chan = \gchan{\proc}{\proc'}$'s contents are the fixed sequence of messages that been sent by $\proc$ but yet to be received by $\proc'$, in the order they were sent.
Thus, the channel has a length $\chan.\len \in \N$, i.e., the number of messages in the channel, and contents $\chan.\contents = \p{\msg_1,\ldots,\msg_{\chan.\len}}$.
The classical and quantum parts of the channel are $\classical_\chan \triangleq (\classical_{\msg_1},\ldots,\classical_{\msg_{\chan.\len}})$ and $\quantum_\chan \triangleq (\quantum_{\msg_1},\ldots,\quantum_{\msg_{\chan.\len}})$, respectively.
Thus, the \emph{channel's classical state} is simply the sequence of classical states of the messages in the channel: $\cstate_\chan = \p{\cstate_{\msg_1},\ldots,\cstate_{\msg_{\chan.\len}}}$, and the \emph{channel's Hilbert space} is the tensor product of the message Hilbert spaces, i.e., $\H_\chan = \bigotimes_{\msg \in \chan.\contents} \H_\msg$.

\paragraph{System State}
The distributed system $\Sys$ is comprised of its processes $\Procs$ and channels $\Chans$.
So, the \emph{classical state} of the system is $\cstate_{\Sys} = \p{\cstate_{\comp}}_{\comp \in \Procs \cup \Chans}$, and the \emph{system's Hilbert space} is $\Hcal_{\Scal} = \bigotimes_{\comp \in \Procs \cup \Chans} \H_\comp$.
Consequently, the system's classical state is $\cstate_\Sys = \p{\cstate_\comp}_{\comp \in \Procs \cup \Chans}$.
The \emph{system's quantum state} is expressed as a single \emph{density matrix} $\rho_{\Scal}$ over $\Hcal_{\Scal}$.
While the global quantum state of the system is not a simple combination of local quantum states of the components of the system (just as a joint distribution cannot be fully described by its marginals), the \emph{local quantum states} of each component $\comp$ is well defined, and can be calculated from $\qstate_\Sys$ by taking a partial trace over the complement of the component with respect to the system.
That is, for a component $\comp$, its quantum state is $\qstate_\comp = \tr_{\overline{\comp}}[\qstate_\Sys]$).
The \emph{state} of the system is $\fstate_\Sys \triangleq (\cstate_\Sys, \qstate_\Sys)$.

\subsection{Operations, Algorithms, Transition Systems, and Executions}

\begin{definition}[operations]
An \emph{operation} $\op_\comp$ is expressed as a mapping $\cstate_\comp \mapsto (\cop_\comp[\cstate_\comp], \qop_\comp[\cstate_\comp])$, where the first component is a quantum operation that is to be applied to $\quantum_\comp$ to change the quantum state and obtain a measurement outcome $r_\comp$ (which may be $\bot$ if no measurement was performed), and the second component is a classical operation that looks at the measurement outcome to determine a new classical state of $\classical_\comp$, i.e., $\cop_\comp(\cstate_\comp, r_\comp) \mapsto \cstate'_\comp$.

An operation $\op_\comp$ is \emph{local} to $\comp$ since it directly touches only $\classical_\comp$ and $\quantum_\comp$ but acts as identity on the complement of $\comp$. Generally, it is only physically permissible for a processor $\proc$ to affect a system by applying operations that are local to $\proc$. We typically refer to operations that act only on a single $\proc$ and its adjacent channel registers as ``local'', and expect that a physically permissible distributed algorithm should consist only of such local operations. In contrast, we use the term \emph{global operation} to emphasize when an operation may not be local to any processor. These operations will be used in specifications, but will not appear in physically permissible algorithms.

\end{definition}

\paragraph{Events.}
The definition of an event in a quantum system requires some thought, as quantum operations involving measurement are inherently probabilistic and produce random outcomes. In this work, we take the point of view that an event is something that \emph{definitely happens}, and the sequence of what events occurs should constitute the total public information available to the adversary. Thus, when an operation corresponding to a measurement is applied, the corresponding event will specify which outcome is obtained. As an example, an operation could be to measure a qubit in the standard basis, while an event could be to measure that qubit in the standard basis \emph{and obtain} the outcome $1$. With this in mind, we formally define the various types of events that may occur in the following definitions.
\begin{definition}[operation application]
    Suppose we apply operation $\op_\comp$ to a component whose classical state is $\classical_\comp$, and obtain the measurement outcome $r_\comp$. This constitutes an \emph{application event} $e_\comp = \Lambda_\comp^{r_\comp}$, where $\Lambda_\comp^{r_\comp}$ corresponds to applying $\Lambda_\comp[\classical_\comp]$ and obtaining result $r_\comp$, as defined in Definition~\ref{def:quantum-operation}. 
\end{definition}

\begin{definition}[message sending]
    We denote the event of $\proc$ sending the message $\msg$ to $\proc'$ by $\send(\proc, \msg, \proc')$.
    We call $\proc$ the \emph{sender} and $\proc'$ the \emph{intended recipient}.
    The quantum part $\quantum_\msg$ of the message to be sent must be a tensor factor of $\quantum_\proc$, i.e., $\quantum_\proc = \quantum_a \otimes \quantum_\msg$. The action of sending consists of relabeling this tensor factor $\quantum_\msg$ so that it is part of $\chan$ rather than $\proc$, \emph{without} changing the global quantum state\footnote{One could imagine simply physically transporting the qubits corresponding to the $\quantum_\msg$ part of $\proc$ out of the processor.}, and setting the classical part $\classical_\msg$ to some value $\sigma_\msg$.
    Letting the contents of the channel $\chan = \gchan{\proc}{\proc'}$ be $\chan.\contents = (\msg_1,\ldots,\msg_k)$, the sending effects the following change:
    $\quantum_\proc \gets \quantum_a$, $\classical_\msg \gets \sigma_\msg$, and $\chan.\contents \gets (\msg_1,\ldots,\msg_k, \msg)$.
\end{definition}

\begin{definition}[message reception]
    We denote the event of $\proc$ receiving the message $\msg$ from $\proc'$ by $\receive(\proc, \msg, \proc')$. We call $\proc$ the \emph{recipient} and $\proc'$ the \emph{sender}. 
    The action of reception consists of popping the first message $\msg$ from $\chan = \gchan{\proc}{\proc'}$, and delivering the classical and quantum parts of $\msg$ to $\proc$. The delivery of the quantum part of $\msg$ consists of relabeling the tensor factor $\quantum_\msg$ so that it becomes part of $\quantum_\proc$, without changing the global quantum state.
    Thus, letting the contents $\chan$ before the reception be $\chan.\contents = (\msg_1,\ldots,\msg_k)$, the reception changes $\chan$'s contents to $\chan.\contents \gets (\msg_2,\ldots,\msg_k)$ and the quantum part of $\proc$ to $\quantum_\proc \gets \quantum_\proc \otimes \quantum_{\msg_1}$, and records $\classical_\mu$ in $\classical_\pi$. 
\end{definition}

We model the observable interactions of the distributed system with the outside world through an interface of \emph{invocations} and corresponding \emph{responses}.

\begin{definition}[invocation and response]
The \emph{invocation} of a an operation $\op$ on process $\proc$ is modeled by the event $\invoke(\proc, \op)$.
Upon receiving an invocation, a process may eventually \emph{respond} to the invocation with a response $\res$, modeled by the event $\respond(\proc, \res)$.
\end{definition}

\begin{definition}[events]
The \emph{events} of a system are the significant happenings in the distributed system.
An event can correspond to either (1) an invocation, (2) a response, (3) the application of an operation, (4) the sending of a message, or (5) the reception of a message.
Each event $e$ is \emph{labeled} by the component $\comp$ (which is a processor in physically permissible executions) at which it occurs, denoted $e.\lab = \comp$.
\end{definition}

\begin{definition}[step]
    A \emph{step} is a triple $(\fstate, e, \fstate')$, where $\fstate$ is a \emph{pre-event state} $\fstate$, and $\fstate'$ is the \emph{post-event state} produced by the occurrence of $e$ in state $\fstate$. A step is \emph{valid} if the post-event state $\fstate'$ is the correct state that would result from $e$ occurring at $\fstate$.
\end{definition}

\begin{definition}[transition system]
    A \emph{transition system} $(\fstate^0, \delta)$ is defined by an \emph{initial state} $\fstate^0$, and a \emph{transition function},
    $\delta: \CStates \times \Events \times \CStates \to \set{\true, \false}$. We interpret $\delta(\cstate, e, \cstate') = \true$ to mean that it is possible in a system whose classical state is $\cstate$ for event $e$ to occur, transitioning the classical state of the system to $\cstate'$, as well as modifying the quantum state of the system. Abusing notation, we sometimes think of $\delta$ as a predicate on steps, i.e. a map $\FStates \times \Events \times \FStates \to \{\true,\false\}$, where the predicate is true if $\delta$ allows the transition between classical states, and the quantum states evolve according to $e$.
\end{definition}
We remark that a transition system is defined this way in terms of the \emph{classical} states to model the fact that in a distributed system, the decision of what event should happen next is only allowed to depend on classical information. In particular, the quantum state cannot be accessed except through measurement. (This is related to the fact that we do not allow for coherent superpositions of execution histories.)

\begin{definition}[execution]
    An \emph{execution} is a triple $X = (\fstate^0, E, <_X)$, where $\fstate^0$ is the \emph{initial state}, $E$ is a set of events, and $<_X$ is a total order over the events of the system that sequences the events into $\vec E = e^1,e^2,\ldots$. 
    The execution is \emph{finite} or \emph{infinite} depending on the cardinality of $E$.
    Given an execution, we can uniquely define, inductively, the states $\fstate^i$, such that $\p{\fstate^{i-1}, e^i, \fstate^i}$ are valid steps.
    We call the sequence $\vec \fstate = \fstate^0,\fstate^1,\ldots$ the \emph{sequence of states} of $X$.
    We often use $\Psi^\X$ to refer to the final state of a finite execution. An execution is \emph{well formed} if messages satisfy at most-once semantics and FIFO, and if each event is compatible with the state just before it takes place.
   For a given transition function $\delta$, an execution $X$ is \emph{an execution of $\delta$} if each step $(\fstate^{i-1}, e^i, \fstate^i)$ is \emph{$\delta$-valid}, in the sense that $\delta(\fstate^{i-1}, e^i, \fstate^i) = \true$.
   
\end{definition}

Since the ordering $<_X$ in an execution is a total order, an equivalent way to specify an execution is simply by an initial state and a sequence of events
\[ X = (\fstate^0, \vec{E}).\]
We will occasionally use this more compact notation where convenient. It is especially useful for discussing fragments of executions and their concatenations.

\begin{definition}
    An \emph{execution fragment} $X$ is an execution that is viewed as a contiguous subsequence of some larger execution $Y$. Specifically, if
    \[ Y = (\fstate^0, (e^1, \dots, e^n))\]
    is an execution, then for any $1 \leq i\leq j \leq n$, we can define an associated fragment
    \[ X_{i:j} = (\fstate^{i-1}, (e^i, \dots e^j)).\]
    The \emph{concatenation} of two execution fragments $X_1, X_2$ is defined in the natural way and denoted using the notation $X_1 \cc X_2$, and is only well defined if the final state of the first fragment is equal to the initial state of the second fragment. Thus, the execution $Y$ could be written as
    \[ Y = X_{1:i} \cc X_{i+1:n}\]
    for any $i \in [n]$.
\end{definition}

\begin{definition}[algorithm]
A \emph{distributed algorithm} $\A$ is a collection of algorithms for each individual processor $\set{\A_\proc}_{\proc \in \Procs}$.
Each $\A_\proc$ specifies:
\begin{itemize}
    \item 
    a set of possible \emph{initial states} $\FStates^0$
    \item
    A \emph{local transition predicate} $\delta_\proc: \CStates \times \Events \times \CStates \to \{\true, \false\}$  such that $\delta_\proc$ is a function only of the classical states of the components local to $\proc$ (i.e. $\proc$ and any adjacent channels).
\end{itemize}
The associated transition system $\delta$ with $\Acal$ is simply the predicate 
\[ \delta = \bigvee_{\proc \in \Procs} \delta_\proc. \]
\end{definition}

\section{Computational Causality}
\label{sec:computational-causality}

For classical distributed systems, Lamport defined the precedence relation of \emph{computational causality} between events of a distributed system \cite{LamportTimeClocksOrderingOfEvents}.
This relation is defined as the transitive closure of the relation that (1) totally orders events at a given processor chronologically, and (2) orders the event of message reception after the corresponding send event.
Lamport showed that since all interactions between different processors of a distributed system are via message transmissions, chronological orderings between events that are unrelated by the computational causal order are essentially spurious.
In quantum distributed systems, due to entanglement, interactions between different processors can be instantaneous even in an asynchronous distributed system.
In particular, these quantum interactions are transmitted faster than the transmission speeds of the message passing channels.
This puts into question the significance of the computational causality relation in quantum distributed systems.

In this section, we define the computational causality relation for quantum distributed systems and prove that even as quantum interactions affect processors and in-flight messages at rates faster than the speeds of channel transmission, the essential equivalence of executions that share a causal skeleton is retained in quantum distributed systems. 

The goal of the computational causality definition below is to capture only the relevant precedence relationships between steps in an execution.

\begin{definition}[Computational causality]\label{def:causality}
    Let $X = (\fstate^0, E, <_X)$ be an execution.
    We can specify the computational causal relation $\cprec_X$ as the transitive closure of the relation containing the following \emph{primitive relationships} between events $e, e' \in E$:
    \begin{itemize}
        \item[A1:] 
        If $e$ and $e'$ are consecutive events on the same process with $e <_X e'$, then $e \cprec_X e'$. 
        \item[A2:]
        If $e $ and $e'$ are, respectively, the sending and reception of the same message, then $e \cprec_X e'$.
    \end{itemize}
    We call $\cprec_X$ the \emph{causal precedence relation} of execution $X$.
\end{definition}
Notice that if $e \cprec_X e'$, then it always holds that $e <_X e'$: this is obvious in the first case, and in the second case, a message can only be received after it is sent.

\begin{definition}[equicausal executions]
    Let $X = (\fstate^0, E, <_X)$ and $\Xprime = (\fstate^0, E, <_\Xprime)$ be two executions that share an initial state and set of events (but may be ordered differently).
    We say that $X$ and $\Xprime$ are \emph{equicausal}, denoted $\X \ceq \Xprime$, if they share a causal precedence relation, i.e., $\cprec_\X = \cprec_\Xprime$.    
\end{definition}

\begin{lemma}\label{lem:substitution}
    Let $\X = \X_- \cc \X_0 \cc \X_+$ be an execution with causal precedence relation $\cprec_\X$, consisting of three fragments $\X_-, \X_0, \X_+$. Let $\Y_0$ be an execution fragment that is equicausal with $\X_0$ and which has the same final state. Then $\Y = \X_- \cc \Y_0 \cc \X_+$ is equicausal with $\X$ and has the same final state as $\X$. Moreover, if $\X, \Y_0$ are execution fragments of a transition system $\delta$, then so is $\Y$.
\end{lemma}
\begin{proof}
    It is clear that $Y$ has the same final state as $\X$, since by assumption $\Y_0$ has the same final state as $\X_0$, and this will feed into $\X_+$ in both executions to yield the same final state.

    Next, for the ``moreover'': if $\X$ and $\Y_0$ are both execution fragments of $\delta$, then this means that $\X_{-}, \Y_0$, and $\X_{+}$ are all execution fragments of $\delta$, and thus $\Y$ is as well, since it's the concatenation of these execution fragments.

    Finally, we show the equivalence, $\cprec_X = \cprec_Y$.
    The relative order of events at any given processor is identical in $<_\X$ and $<_\Y$, since the reordering $\Y_0$ of $\X_0$ respects $\cprec_\X$.
    For each message $m$ that is sent at event $e$ and received at event $e'$, $e <_\Y e'$ since:
    (1) if $e$ and $e'$ occur in different fragments or they both occur in the same fragment $\X_-$ or $\X_+$, then this follows since $e <_X e'$, and
    (2) if $e$ and $e'$ both occur in $\X_0$, then this follows since $\Y_0$ is equicausal to $\X_0$. 
    So, by A1 and A2, $\cprec_X \subseteq \cprec_Y$.
    The proof that $\cprec_Y \subseteq \cprec_X$ is symmetric.
\end{proof}

\begin{fact}\label{fact:commuting-quantum-operations}
    Let $\Hcal = \Hcal_{A} \otimes \Hcal_{B} \otimes \Hcal_{C}$. Suppose $\Lambda^{r_A}_A$ is a quantum operation on $\Hcal$ acting nontrivially only on the $A$ system and $\Lambda^{r_B}_B$ is a quantum operation on $\Hcal$ acting nontrivially only on the $B$ system. Then $\Lambda^{r_A}_A \circ \Lambda^{r_B}_B = \Lambda^{r_B}_B \circ \Lambda^{r_A}_A$.
\end{fact}
\begin{proof}
    If $\Lambda^{r_A}_A$ acts nontrivially only on the $A$ system, then it can be written as
    \[ \Lambda^{r_A}_A = T^{r_A}_A \otimes \mathrm{id}_B \otimes \mathrm{id}_C,\]
    where $T^{r_A}_A$ is a quantum operation on $\Hcal_A$. Likewise,
    \[ \Lambda^{r_B}_B = \mathrm{id}_A \otimes T^{r_B}_B \otimes \mathrm{id}_C,\]
    where $T^{r_B}_B$ is a quantum operation on $\Hcal_B$. Thus,
    \[\Lambda^{r_A}_A \circ \Lambda^{r_B}_B = \Lambda^{r_B}_B \circ \Lambda^{r_A}_A = T^{r_A}_A \otimes T^{r_B}_B \otimes \mathrm{id}_C. \]    
\end{proof}

\begin{lemma}\label{lem:inversion-swap}
    Let $\X = \p{\fstate, \set{e_1, e_2}, \set{e_1 <_\X e_2}}$ be an execution fragment of some algorithm $\Acal$, with $e_1 <_\X e_2$. Suppose $e_1 \ncprec_X e_2$, and consider the reordered execution $\Xprime = \p{\fstate, \set{e_1, e_2}, \set{e_2 <_\Xprime e_1}}$, in in which the two events have been swapped. Then $\Xprime$ is a valid execution fragment of $\Acal$, $\Xprime \ceq X$, and $\Xprime$ and has the same final state as $\X$.
\end{lemma}
\begin{proof}
    We are given that $e_1 \ncprec_X e_2$. From this it follows that $e_1$ and $e_2$ must be on different processors, and moreover, $e_1$ and $e_2$ cannot be the sending and reception of the same message. Let $\pi_1$ be the processor on which $e_1$ acts and $\pi_2$ be the processor on which $e_2$ acts. Then the classical state $\cstate_{\pi_2}$ of $\pi_2$ after $e_1$ in execution $X$ is the same as its value in $\fstate$, and likewise the classical state $\cstate_{\pi_1}$ of $\pi_1$ after $e_2$ in execution $\Xprime$ is the same as its value in $\fstate$. Thus, both transitions in $\Xprime$ are valid for $\Acal$, since they occur in processors with the same classical states as the corresponding transitions in $\X$. This establishes that $\Xprime$ is a valid execution fragment of $\Acal$.
    
    Now, to show that $\X \ceq \Xprime$, observe that in the reordered execution $\Xprime$, it still holds that $e_2 \ncprec_\Xprime e_1$, since the two events are on separate processors and not connected by a message. Moreover, since $e_1 <_X e_2$ and $e_2 <_\Xprime e_1$, it holds that $e_2 \ncprec_\X e_1$ and $e_1 \ncprec_\Xprime e_2$. Thus, both $\X$ and $\Xprime$ have an empty causal precedence relation, and thus $\X \ceq \Xprime$.

    Finally, we must show that the final states of $\X$ and $\Xprime$ are equal. By the above, we have shown that $e_1$ and $e_2$ act on totally disjoint classical and quantum registers of the system. If either $e_1$ or $e_2$ is a message send or reception or invocation or response, then the claim is clearly true, since such an event acts trivially on the quantum part of the system, and the classical operations are on different processors and thus commute.  
        
    Thus, it suffices to consider the case of $e_1, e_2$ both applications of operations. Here too, the conclusion will follow directly from the commutativity of \emph{quantum} operations on distinct quantum tensor factors, but we will spell out the details for completeness. Any application of an operation consists of the application of a deterministic transition function to the classical states of the classical registers it touches, together with a quantum operation on the quantum registers it touches. Specifically, recall from above that $\proc_1$ starts in state $\cstate_{\proc_1}$ and $\proc_2$ starts in state $\cstate_{\proc_2}$ in the initial state $\fstate$; let $\cstate_{rest}$ be the classical state of the rest of the system in $\fstate$, and $\rho_0$ be the quantum state of the system in $\fstate$. Then in both $\X$ and $\X'$, the event $e_1$ consist of $\proc_1$ performing operation
    $(\cop_{\proc_1}[\cstate_{\proc_1}], \qop_{\proc_1}[\cstate_{\proc_1}])$, and the event $e_2$ consists of  $\proc_2$ performing operation
    $(\cop_{\proc_2}[\cstate_{\proc_2}], \qop_{\proc_2}[\cstate_{\proc_2}])$.
    In $\X$ and $\Xprime$ the respective final states are
    \begin{align} \fstate_X &= ((\underbrace{\cop_{\proc_1}(\cstate_{\proc_1}, r_1)}_{\text{state of $\proc_1$}}, \underbrace{\cop_{\proc_2}(\cstate_{\proc_2}, r_2)}_{\text{state of $\proc_2$}}, \underbrace{\sigma_{rest}}_{\text{state of rest}}), \underbrace{\qop_{\pi_2}^{r_2} \circ \qop_{\pi_1}^{r_1} (\rho_0)}_{\text{quantum state}} )\\
\fstate_\Xprime &= ((\underbrace{\cop_{\proc_1}(\cstate_{\proc_1}, r_1)}_{\text{state of $\proc_1$}}, \underbrace{\cop_{\proc_2}(\cstate_{\proc_2}, r_2)}_{\text{state of $\proc_2$}}, \underbrace{\sigma_{rest}}_{\text{state of rest}}), \underbrace{\qop_{\pi_1}^{r_1} \circ \qop_{\pi_2}^{r_2} (\rho_0)}_{\text{quantum state}} ).
    \end{align}
    The classical parts of $\fstate_\X$ and $\fstate_\Xprime$ agree by inspection. For the quantum parts, we see from Fact~\ref{fact:commuting-quantum-operations} that the two quantum operations commute with each other, so the quantum states agree as well. Therefore, $\fstate_\X = \fstate_\Xprime$ as claimed.
\end{proof}

\begin{lemma}\label{lem:move-to-end}
Let $\X = (\fstate^0, E, <_\X)$ be a finite execution. 
Let the total order of the events induced by $<_X$ be $\vec{E}_{<_X} = (e^i)_{i \in [k]}$.
Suppose $e^i \in E$ is an event with no successors under $\cprec_\X$, and $e^j \in E$ be an event such that $e^i <_\X e^j$.
Let $<_\Xprime$ be the relation corresponding to the same total order as $<_\X$, except that $e$ is reordered to occur right after $e'$, i.e., $\vec{E}_{<_\Xprime} = (e^1,\ldots,e^{i-1},e^{i+1},\ldots,e^j,e^i,e^{j+1},\ldots,e^k)$.
Let $\Xprime = (\fstate^0, E, <_\Xprime)$. Then, $\X$ and $\Xprime$ are equicausal and have the same final state.
\end{lemma}
\begin{proof}
    The proof is by induction: we repeatedly apply Lemma~\ref{lem:substitution} with the middle fragment $\X_0$ taken to be $e^i$ and the event immediately after it; we show that swapping $e^i$ forward preserves causality and the final state using Lemma~\ref{lem:inversion-swap}.

    More precisely, let $\X^\ell$ be the execution corresponding to the ordering
    \[\vec{E}_{<_{\X^\ell}} =  (e^1, \dots, e^{i-1}, e^{i+1}, \dots, e^{\ell-1} e^{i}, e^\ell, \dots). \]
    The inductive hypothesis is that $X^\ell$ is equicausal with $X$ and that $\X^\ell$ has the same final state as $\X$. Our base case is $\X^{i+1} = \X$, where both these properties hold.
    
   Now, split the execution $\X^\ell$ as
    \[ X^\ell = X_{-} \cc \underbrace{X_0}_{(e^{i}, e^\ell)} \cc X_{+}. \]
    By hypothesis, $\X^\ell \ceq \X$, so in $\X^\ell$, $e^i$ has no successors under $\cprec_{\X^\ell}$. Therefore, in the middle segment $X_0$, it holds that $e^i \ncprec e^{\ell}$. Therefore, we can apply Lemma~\ref{lem:inversion-swap} to replace $X_0$ with $Y$ corresponding to the order $e^\ell <_\Y e^i$, and then apply Lemma~\ref{lem:substitution} to say that the execution
    \[ X_{-} \cc Y \cc X_{+} \]
    is equicausal with and has the same final state as $\X$. But this is exactly $\X^{\ell+1}$. 
\end{proof}

\begin{definition}[Computational Lightcones]
    Let $\X = (\fstate^0, E, <_\X)$ be an execution. 
    Then for any set of events $D \subseteq E$, its \emph{computational past (resp. future) lightcone} $\Lcal_P(D)$ (resp. $\Lcal_F(D)$) is defined as the set of all events in $E$ that precede (resp. succeed) any event in $D$ under the computational causality order:
    \begin{align}
        \Lcal_P(D) &= \{ e \in E: \exists d \in D, e \cprec_{X} d\} \\
        \Lcal_F(D) &= \{ e \in E: \exists d \in D, d \cprec_{X} e \}.
    \end{align}
\end{definition}

\begin{theorem}[Equicausal executions have equivalent final states]\label{thm:equiv}
    Let $\X$ and $\Xprime$ be finite executions with final states $\fstate^\X$ and $\fstate^\Xprime$, respectively.
    If $\X \ceq \Xprime$, then $\fstate^\X = \fstate^\Xprime$.
\end{theorem}

\begin{proof}
    The proof is by induction over the number of events in the executions.
    For the base case, consider $\X$ that has zero events.
    By equicausality, the initial states of $\X$ and $\Xprime$ are identical, and since they have no events, this is also the final state.

    The inductive hypothesis is that all equicausal executions with $k$ or fewer events have identical final states.
    For the induction step, consider equicausal $\X = (\fstate^0, E, <_\X)$ and $\Xprime = (\fstate^0, E, <_\Xprime)$ with $k+1$ events.
    Let $e$ be the last event in $\X$.
    Let $<_{\Xprimeprime}$ be $<_\Xprime$, just with $e$ reordered to be the last event, and define $\Xprimeprime = (\fstate^0, E, <_{\Xprimeprime})$.
    By Lemma~\ref{lem:move-to-end}, $\Xprime$ and $\Xprime'$ are equicausal and have the same final state.
    By transitivity of equicausality, $\X$ and $\Xprimeprime$ are also equicausal, and by definition, they share the same last event, $e$.
    Let $\Y$ and $\Yprimeprime$ be the prefixes of $\X$ and $\Xprimeprime$, respectively, that are obtained by removing the last event.
    $\Y$ and $\Yprimeprime$ are equicausal and have only $k$ events, thus they share a final state, $\fstate^\Y$.
    Since, $\X$ and $\Xprimeprime$ are both obtained by event $e$ occurring at $\fstate^Y$, they share a final state $\fstate^\X$.
    Since $\Xprime$ and $\Xprimeprime$ are known to share a final state, $\fstate^\X$ is also the final state of $\Xprimeprime$.    
\end{proof}

\section{Specification of Consistent Quantum Global Operations}

Quantum distributed algorithms can only execute atomic operations that are process-local.
Nevertheless, we would like to implement systems which simulate atomic global operations.
Our algorithm, presented in the next section, will be inspired by the global snapshotting algorithm of Chandy and Lamport~\cite{ChandyLamport}.
Snapshots are among the easiest global operations to define, since they only read the state but do not disturb it.
By contrast, the analogous quantum global operation, i.e, the \emph{quantum snapshot}, must measure the quantum state of the entire system, which alters the state of the underlying system.
Such an alteration of the underlying system state is not unique to quantum operations: even classical global operations can be state altering.
As an example, a security application may wish to atomically encrypt the global state of a system and substitute operations with their homomorphic equivalents.
The specification for global operations that we lay out in this section is applicable both to classical and quantum global operations.

\begin{definition}[Decomposable Global Operations]
    A \emph{decomposable} global operation $G$ is an operation such that the following hold.
    \begin{enumerate}
        \item The quantum portion $\Lambda_G$ has outcomes $\vec{r}$ consisting of a tuple of outcomes indexed by $\Procs \cup \Msgs$, and can be written as a tensor product of operations over the components of the system: 
    \[ \Lambda_G^{\vec{r}}[\cstate] = \left(\bigotimes_{\proc \in \Procs} \Lambda_{G, \proc}^{r_\proc}[\cstate_{\proc}] \right) \otimes \left(\bigotimes_{\msg \in \Msgs} \Lambda_{G, \msg}^{r_{\msg}}[\cstate_{\msg}]\right). \]
    \item The classical portion $\cop^G$ factors into maps on the individual components of the system in a similar way:
    \[ \cop^G(\cstate, \vec{r}) = \left( \BigTimes_{\proc \in \Procs} \cop^G_{\proc}(\cstate_{\proc}, r_{\proc}) \right) \times \left( \BigTimes_{\msg \in \Msgs} \cop^{G}_{\msg}(\cstate_{\msg}, r_{\msg})\right). \]
    \end{enumerate}
   The set of all global operations is denoted $\GOps$.
\end{definition}

Our algorithm, presented in the next section, will implement consistent global operations in the case that:
(1) the global operations invoked come from the set $\Gcal$ of \emph{decomposable} operations, i.e., they can be expressed as independently operating on separable components of the system, and
(2) global operations are not invoked concurrently.
The first condition is met by operations such as snapshots, simultaneous global measurement, global encryption, etc.
The latter condition is easily satisfied by ensuring that a single \emph{leader} process takes charge of all global operations.

Let $\A$ be an algorithm on system $\Sys = (\Procs, \Chans)$, with associated transition function $\delta$.

\begin{definition}[atomic specification]\label{def:atomic-spec}
Let $\A$ be an algorithm on system $\Sys = (\Procs, \Chans)$, with associated transition function $\delta$.
The \emph{atomic specification of $\delta$ with global operations in the set $\GOps$} is the transition function $\deltahat$ that is defined as follows:
\begin{itemize}
\item
The \emph{classical state space} of each process $\proc \in \Procs$ is 
$\CStateshat_\proc = \CStates_\proc \times \XCStates_\proc$, where $\XCStates_\proc = \set{\bot} \cup \set{(op, res) \mid \GOps \times \p{\Res \cup \set{\bot}}}$.
\item
The valid transitions under $\deltahat$ correspond to:
\begin{itemize}
    \item 
    {\bf Base Events:}
    
    $\deltahat((\cstate, \xcstatehat), e, (\cstate', \xcstatehat))$, where $(\cstate, e, \cstate')$ is a $\delta$-step.
    
    \item
    {\bf Invocation of Global Operation:}
    
    $\deltahat((\cstate, \xcstatehat), e, (\cstate, \xcstatehat'))$, where 
    $e$ is the invocation of a global operation on a process $\proc$ whose $\xcstatehat_\proc = \bot$ and $\xcstatehat'$ is defined by $\xcstatehat'_{\ol\proc} = \xcstatehat_{\ol\proc}$ for each $\ol\proc \ne \proc$, and $\xcstatehat'_{\proc} = (\gop, \bot)$. 

    \item
    {\bf Atomic Execution of Global Operation:}
    
    $\deltahat((\cstate, \xcstatehat), e, (\cstate', \xcstatehat'))$, where 
    $e$ is the atomic execution of the pending global operation $G$ that is recorded in the state  $\xcstatehat_\proc \ne \bot$ of some leader process $\proc$ (for which $\xcstatehat_\proc = (\gop, \bot)$).
    The final state of this transition is as follows:
    \begin{itemize}
        \item For each processor and message, the associated state $\cstate'$ is the result of the classical portion of the global operation $G$ applied to that component.
        \item For each processor $\proc$, $\xcstatehat_{\proc}' = (G, R_{\proc})$ where $R_{\proc}$ records the outcomes of the global operation $G$ on that processor $\proc$ and on all messages that are currently in flight to $\proc$ but have not been received.
    \end{itemize}
    \item
    {\bf Response to a Global Operation:}
    
    $\deltahat((\cstate, \xcstatehat), e, (\cstate, \xcstatehat'))$, where 
    $e$ is the event that some processor $\proc$ responds with $R_{\proc}$, $\xcstatehat_{\proc} = (G, R_{\proc})$ contains a record of $R_{\proc}$, and $\xcstatehat_{\proc}'$ is reset to $\bot$. For all other processors $\proc' \neq \proc$, $\xcstatehat'_{\proc'} = \xcstatehat_{\proc}$.

\end{itemize}
\end{itemize}
\end{definition}

\begin{remark}
The definition above is general, allowing processes to invoke global operations concurrently.
However, like Chandy and Lamport's snapshot algorithm, in this paper, we will only be interested in algorithms that \emph{do not} allow for concurrent invocations of global operations (even by different processes).
\end{remark}

\section{Quantum Global Operations Algorithm}

The Quantum Global Operation (QGO) Algorithm $\At$ is defined in terms of a base algorithm $\A$, which can be any quantum distributed algorithm, and a decomposable global operation $G \in \GOps$. The algorithm is closely modeled on the Chandy-Lamport algorithm~\cite{ChandyLamport}---in fact, all the steps are essentially the same.

$\At$ consists of $\A$ augmented in the following ways. Firstly, each component (processor or message) gets an additional classical register $\xcstatet$, which will be used by the algorithm to store information related to the global operation. Specifically, for a processor $\proc$, $\xcstatet_{\proc}.op$ will store the global operation currently underway, or $\emptyset$ if no operation is underway;  $\xcstatet_{\proc}.res$ contains a response that is under construction; and $\xcstatet_{\proc}.waitset$ contains a list of channels that the processor $\proc$ is waiting to hear from. For each message, $\xcstatet_{\msg}$ will be used to store a \emph{marker} indicating the global operation $G$ to be performed, and also to temporarily store the outcome of the global operation on the message before it is transferred to $\xcstatet_{\proc}.res$ register of the recipient $\proc$. 

$\At$ adds to $\A$ three procedures, which are described in pseudocode in Algorithm~\ref{alg:quantum-global-operation-algorithm}. 
Before describing them in detail, we sketch the action of the algorithm at a high level. 
When a new global operation $G$ is invoked on a process $\proc$, it applies $G.\proc$ to its own state, records the outcome, and starts an empty record for every incoming channel: these records are meant to hold messages that are currently ``in flight'', and $\proc$ will apply the global operation $G.\msg$ to any message $\msg$ received hereafter and record the results, until the stopping condition (to be described) is achieved. 
Process $\proc$ then broadcasts a ``marker message'' to all other processes. Each process that receives a marker behaves exactly as the $\proc$ did, when it receives the marker for the first time. 
Whenever a process (including the process on which the global operation was invoked) receives subsequent markers, it marks the record of the corresponding channel as complete and stops recording further messages on that channel. Finally, whenever a process has marked all channels as complete, it responds with its complete record, consisting of the outcome it obtained from its own state as well as all recorded message outcomes. For simplicity, we allow each process to respond separately, and do not require the algorithm to aggregate the responses---the algorithm can easily be modified to do this by having processes send their responses back to the invoking process.

Now, in more detail, we describe the three added procedures in $\At$.
    \begin{itemize}
    \item   $\At.\Invoke$: this takes as an argument a description of a global operation $G$, and is executed by the leader process $\proc$ when it receives an invocation of the QGO algorithm on $G$. This initiates the global operation by calling $\At.\ProcessNewGlobalOp$ on $\proc$, with the second argument set to $\bot$ to indicate that this is an invocation received from outside the system and not caused by a message from another processor in the system.
    \item $\At.\ProcessNewGlobalOp$: this takes two arguments $G, \chan$ indicating the global operation and the channel on which the marker was received (if any). This perform the steps described above: the processor's local state is acted on with $G$, and the result is recorded. Records are then opened for each incoming channel by initializing an empty record in $\xcstatet_{\proc}.res$ and adding the channel to $\xcstatet_{\proc}.waitset$, and finally, the processor sends out a marker, i.e. a message with $\xcstatet_{\msg} = G$, to all other processors.
    
    \item $\At.\Receive$: this takes one argument $\msg$, which is a message that has been received. We imagine this procedure as ``intercepting'' all message receptions before they reach the base algorithm $\A$. Upon receipt of the message, the processor decides if it is a marker message, by checking whether $\xcstatet_{\msg}$ is nonempty. If so, and if this is the first time the marker is being received, then the processor calls $\At.\ProcessNewGlobalOp$ on itself. Otherwise, the processor marks the channel $\chan$ on which this marker was received as closed, by removing it from $\xcstatet_{\proc}.waitset$, and if all the waitsets are empty, it sends out its response.

    If this is not a marker message, but it is received on a channel $\chan$ that the processor is waiting for, then the global operation is applied to the message and the results are recorded in $\xcstatet_{\proc}.res.\chan$. 
    
    Finally, in all cases, the message is passed on to the base algorithm $\A$.
    \end{itemize}

    Finally, we remark on the relation of this algorithm to the events of the system. We assume that each procedure in $\At$ is executed \emph{atomically}, meaning that no other events can occur on the processor $\proc$ while it is executing $\Invoke$, $\ProcessNewGlobalOp$, or $\Receive$. Each procedure execution can correspond to multiple events, and the line numbers in the pseudocode in Algorithm~\ref{alg:quantum-global-operation-algorithm} indicate where a new event starts.

\begin{algorithm}[h]

\begin{algorithmic}[1]

\State $\At.\Invoke_\pi(G)$:
\Statex\tab[2] $\ProcessNewGlobalOp_{\proc}(G, \bot)$

\bigskip

\State $\At.\ProcessNewGlobalOp_\proc(G, \chan)$ 
\Statex\tab[2]       $\xcstatet_\proc.\op = G$
\State\tab[2]        $\xcstatet_\proc.\resp.\proc \gets \Execute_\proc(G.\proc)$
\Statex\tab[2]       $\xcstatet_\proc.\waitset \gets \Chans_{* \to \proc}$
\Statex\tab[2]       \cmdif $\chan \ne \bot$ \cmdthen 
\Statex\tab[3]          $\xcstatet_\proc.\resp.\chan \gets \varnothing$
\Statex\tab[3]          $\xcstatet_\proc.\waitset \gets \xcstatet_\proc.\waitset \setminus \set{\chan}$
\Statex\tab[2]       $\msg \gets \NewMessage[\xcstatet_\msg = G]$ 
\State\tab[2]        $\Broadcast_{\Chans_{\proc \to *}}(\mu)$

\bigskip

\State $\At.\Receive_{\pi,\chan}(\msg)$
\Statex\tab[2] \cmdif $\xcstatet_\mu \ne \varnothing$ \cmdthen
\Statex\tab[3]   \cmdif $\xcstatet_\proc = \varnothing$ \cmdthen \Comment{New global operation} 
\Statex\tab[4]      $\ProcessNewGlobalOp_{\proc}(\xcstatet_\msg, \chan)$
\Statex\tab[3]   \cmdelse \Comment{Ongoing global operation} 
\Statex\tab[4]       $\xcstatet_\proc.\waitset \gets \p{\xcstatet_\proc.\waitset \setminus \set{\chan}}$
\Statex\tab[4]       \cmdif $\tau_\pi.\waitset = \nullset$ \cmdthen 
\Statex\tab[5]          $\Respond(\xcstatehat_\proc.\resp)$
\Statex\tab[5]          $\xcstatehat_\proc.\resp \gets \varnothing$ 
\Statex\tab[5]          $\xcstatet_\proc \gets \varnothing$
\Statex\tab[2] \cmdelse \Comment{This is a regular message, not a marker}
\Statex\tab[3]   \cmdif $\xcstatet_\proc \ne \varnothing$ \cmdand $\chan \in \xcstatet_\proc.\waitset$ \cmdthen \Comment{If a global operation is underway}
\State\tab[4]       $\xcstatet_\msg \gets \Execute_{\msg}(\xcstatet_\proc.\op.\msg)$
\State\tab[4]       $\p{\xcstatehat_\proc.\resp.\chan}.\Append(\xcstatet_\msg)$
\Statex\tab[4]       $\xcstatet_\msg \gets \varnothing$
\Statex\tab[3]       $\A.\Receive_{\proc, \chan}(\msg)$ \Comment{Receive $\msg$ according to the original algorithm}

\end{algorithmic}

\caption{Pseudocode description of the {\em Quantum Global Operation Algorithm}.}

\label{alg:quantum-global-operation-algorithm}
\end{algorithm}

\section{Analysis of the algorithm}
\label{sec:proof}

We will now show that the QGO algorithm is a correct implementation of the specification of an atomic global operation given in Definition~\ref{def:atomic-spec}. To state our theorem, we will first need to define the notion of a \emph{history} of an execution. 
\begin{definition}
    Let $F$ be some set of events, called the \emph{filter}. Then for any execution $X = (\fstate^0, \vec{E})$, the \emph{history} of that execution filtered through $F$ is denoted $\Hist_F(X)$, and consists of the linear sequence of events in $\vec{E}$ that are contained in $F$.
\end{definition}
In the context our algorithm, we will take the filter set $F$ to be the set of events that are either invocations, responses, or operations of the base algorithm $\A$. The idea is that a history consists of the steps in an execution that have observable external effects that we care about. In what follows, we will fix this filter set and suprress the $F$ subscript, merely writing $\Hist(X)$ for the history. The next definition allows us to compare histories of our algorithm $\At$ with the specification $\deltahat$.

\begin{definition}
    For an execution $\Xt = (\fstate, \vec{\Et} = (e^1, e^2, \dots))$ of $\At$ and an execution $\Xhat = (\fstatehat, \vec{\Ehat} = (\wh{e}^1,  \wh{e}^2, \dots)) $ of $\deltahat$, we say that event $e \in \vec{E}$ and $\wh{e} \in \vec{\Ehat}$ satisfy $e \simeq \wh{e}$ if both events are in $F$, and $e$ and $\wh{e}$ correspond to the same action: either (1) $e$ are $\wh{e}$ are both invocations on the same processor $\proc$ of the same global operation $G$, or (2) $e$ and $\wh{e}$ are response events on the same processor that give the same response, or (3) $e$ and $\wh{e}$ are base events corresponding to the same step in the algorithm $\A$.
    
    We say that
    \[ \Hist(\X) \simeq \Hist(\Xhat) \]
    if $|\vec{E}| = |\vec{\Ehat}|$, and for every $i$, $e^i \simeq \wh{e}^i$.
\end{definition}

Armed with these definitions, we can now state our theorem about the correctness of our algorithm $\At$. 

\begin{theorem}\label{thm:main}
Consider any finite execution $\Xt = (\fstatet^0, \Et, <_\Xt)$ of $\At$ that has no pending global operation and no concurrent invocations of global operations (i.e every invocation happens after all responses from the previous invocation have been given).
There is an execution $\Yt$ of $\At$ and a specification execution $\Yhat$ of $\deltahat$ such that $\Xt \ceq \Yt$ and $\Hist(\Yt) \simeq \Hist(\Yhat)$.
\end{theorem}
Before we prove the theorem, let us say a few words about why this is the right correctness condition for the algorithm. The basic message of Chandy and Lamport's original analysis~\cite{ChandyLamport} is that the global snapshot obtained by their algorithm in execution $X$ is not guaranteed to be equal to the global state at any given point in time in $X$, but there is still guaranteed to exist an execution $Y$ that has the same starting and final state as $X$, and such that the global state of $Y$ at some point in time matches the output of the CL algorithm. In our theorem, we show a slightly stronger guarantee: given an execution $\Xt$ of our algorithm, we show that there exists an execution $\Yt$ that shares not just its initial and final state with $\Xt$ but is causally equivalent to $\Xt$ (recall that by the definition $\ceq$, this implies that $\Yt$ and $\Xt$ share the same initial state, and by Theorem~\ref{thm:equiv}, this implies that the final states of $\Yt$ and $\Xt$ match). The specification execution $\Yhat$ in our theorem is a means of formalizing the claim that the output of our snapshot algorithm matches the global state of some point in $\Yt$, in a way that generalizes to other types of global operations beyond measurement.

\begin{proof}[Proof of Theorem~\ref{thm:main}]
Since global operations are not concurrent, we can decompose the execution $\Xt$ into the following concatenation, where the $\Xt^i_{main}$ execution fragments correspond to the segment between the $i$th invocation and the last response it generates.
\[\Xt = (\fstatet^0, \Xt^0 \cc \Xt^1_{main} \cc \Xt^1 \cc \Xt^2_{main} \cc \cdots \cc \Xt^k_{main} \cc \Xt^k)\]

\paragraph{Constructing $\Yt$ by reordering processor operations:} Fix an arbitrary $i \in [1,k]$, and consider $\Xt_{main} = \Xt^i_{main}$ with event set $\Et_{main}$.
Our strategy will be to reorder $\Xt_{main}$ into an alternate execution $\Yt_{main}$ while preserving its essence.

First, we tripartition the events $\Et_{main}$ of $\Xt_{main}$ into $\Et_{pre}$, $\Et_{op}$, and $\Et_{post}$, as we describe below. For each processor $\proc \in \Procs$, define $\Et_{main, \proc}$ to be the set of events in $\Et_{main}$ occurring at processor $\proc$.
We note that each processor $\proc \in \Procs$ executes $G.\proc$ exactly once in $\Et_{main}$;
let $e^{op}_\proc \in \Et_{main}$ be the event at which it does so.
Also, let $e^r_\proc$ be the event at which $\proc$ first received a marker message indicating this global operation---with $e^r_{\proct} = e_i$---and let $\Et^s_\proc$ be the send events corresponding to $\proc$ broadcasting the marker.
By definition of the algorithm, for each $\proc \in \Procs$: $e^r_\proc$ immediately precedes $e^{op}_\proc$, which immediately precedes the broadcast $\Et^s_\proc$.
For the remainder of the argument, we identify these events together as $e^*_\proc$.
$\Et_{pre,\proc} \triangleq \set{e \in \Et_{main, \proc} \mid e <_{\Xt} e^*_\proc}$,
$\Et_{post,\proc} \triangleq \set{e \in \Et_{main, \proc} \mid e >_{\Xt} e^*_\proc}$, and
$\Et_{op,\proc} \triangleq \Et_{main, \proc} \setminus \p{\Et_{pre,\pi} \cup \Et_{post,\pi}}$.
Finally, define
$\Et_{pre} \triangleq \bigcup_{\proc \in \Procs} \Et_{pre,\proc}$, 
$\Et_{op} \triangleq \bigcup_{\proc \in \Procs} \Et_{op,\proc}$, and 
$\Et_{post} \triangleq \bigcup_{\proc \in \Procs} \Et_{post,\proc}$.
We call these three sets of events, \emph{pre-events}, \emph{op-events}, and \emph{post-events}. Note that the \emph{op-events} as defined here only include the global operations as applied to the processors, \emph{not} the operations on the messages (which will be in the post-events).

Note that since the events at each processor were partitioned independently, it is possible that post-events at one processor precede pre-events at another processor.
Our goal is to re-order the execution fragment $\Xt_{main}$ in an execution $\Yt_{main}$, such that all pre-events precede all op-events and all op-events precede all post-events in $\Yt_{main}$.
We do this inductively by eliminating \emph{inversions}, as is justified by proving the following claim.

\begin{claim} \label{clm:pre-op-post}
Consider an execution fragment $\Ut = \p{\fstatet_{pre}, \Et_{main}, <_\Ut}$ of $\At$ in which events $e_1$ and $e_0$ are consecutive with $e_1 <_\Ut e_0$.
If either: 
(a) $e_1 \in \Et_{post}$ and $e_0 \in \Et_{pre}$,
(b) $e_1 \in \Et_{post}$ and $e_0 \in \Et_{op}$, or
(c) $e_1 \in \Et_{op}$ and $e_0 \in \Et_{pre}$; then 
(1) $e_1 \ncprec e_0$, and thus 
(2) if we define $<_\Vt$ to be identical to $<_\Ut$ except that $e_0 <_\Vt e_1$, then $\Vt = \p{\fstatet_{pre}, \Et_{main}, <_\Vt} \ceq \Ut$ and $\Vt$ is a valid execution fragment of $\At$.
\end{claim}

\begin{proof}
We prove conclusion (1) by contradiction.

Let $\proc_0$ and $\proc_1$ be the processes at which events $e_0$ and $e_1$ occur, respectively.
By the definition of the tripartition of events at each given process $\proc$ into pre-, atomic-, and post-events, it is clear that $\proc_0 \ne \proc_1$ if any of (a), (b), or (c) holds.
Thus, by the definition of $\cprec$ and the fact that these two events $e_1$ and $e_0$ are adjacent in the total ordering $<_Y$, $e_1$ must be the sending of a message $m$, which is received at $e_0$.
We conclude the proof in cases based on which hypothesis holds.
\begin{itemize}
    \item 
    \ul{Case (a): $e_1 \in \Et_{post}$ and $e_0 \in \Et_{pre}$:}

    All marker send events are in $\Et_{op}$, thus the message sent in $e_1$ must be a non-marker message.
    By the definition of post-events, $e^*_{\pi_1} <_\Ut e_1$, and by the definition of pre-events, $e_0 <_\Ut e^*_{\pi_0}$.
    However, by the FIFO property of the channel $\gchan{\pi_1}{\pi_0}$, any message sent by $\pi_1$ after $e^*_{\pi_1}$ must be received by $\pi_0$ after $e^*_{\pi_0}$---since $e^*_{\pi_0}$ corresponds with $\pi_0$'s first receipt of the marker.
    Consequently, we conclude that $e^*_{\pi_0} <_\Ut e_0$, which is a contradiction.

    \item 
    \ul{Case (b) $e_1 \in \Et_{post}$ and $e_0 \in \Et_{op}$:}

    The only receive events in $\Et_{op}$ are of the marker.
    Thus $e_1$ must be the sending of the marker by $\pi_1$ to $\pi_0$ and $e_0$ must be its reception.
    However, markers are only sent in events in $\Et_{op}$, which is a contradiction.

    \item
    \ul{Case (c) $e_1 \in \Et_{op}$ and $e_0 \in \Et_{pre}$:}

    The only send events in $\Et_{op}$ are of the marker.
    However, by the definition of pre-events, $\pi_0$ cannot receive the marker before $e^*_{\pi_0}$---since that event is the first receipt of the marker by $\pi_0$.
    Thus, we reach a contradiction.
\end{itemize}
That finishes the proof of (1).
Since $e_1 \ncprec e_0$, conclusion (2) follows by Lemma~\ref{lem:inversion-swap}.
\end{proof}

Since $\Xt_{main}$ is a finite run, we can induct on the claim to remove all inversions and conclude that there is a re-ordering of $\Yt_{main}$ of $\Xt_{main}$, such that $\Yt_{main} \ceq \Xt_{main}$ and all pre-events precede all op-events and all op-events precede all post-events in $\Yt_{main}$. In other words, there exist execution fragments $\Yt_{pre}, \Yt_{op}, \Yt_{post}$ consisting solely of events of the appropriate type, such that
\[ \Yt_{main} = \Yt_{pre} \cc \Yt_{op} \cc \Yt_{post}.\]
Thus, by repeated application of Lemma~\ref{lem:substitution} we can substitute each $\Xt^i_{main}$ with $\Yt^i_{main}$ to get 
\[\Yt = (\fstatet^0, \Xt^0 \cc \Yt^1_{main} \cc \Xt^1 \cc \Yt^2_{main} \cc \cdots \cc \Yt^k_{main} \cc \Xt^k) \ceq \Xt. \]

\paragraph{Constructing $\Zt$ by reordering message operations:}
\[ \Yt^i_{main} = \Yt^i_{pre} \cc \Yt^i_{op} \cc \Yt^i_{post} \]

We are going to replace the fragment $\Yt^i_{op} \cc \Yt^i_{post}$ by a reordered fragment $\Zt^i_{op,post}$. 
Let $\fstate^{pre}$ be the initial state of $\Yt_{op}$, and let
\[ \Yt^i_{op,post} = \Yt^i_{op} \cc \Yt^i_{post}.\] 
Let $\vec M$ be the sequence of all messages $\msg$ with a $e^{op}_\msg$ event in $\Yt^i_{main}$. 
Let $\vec{\Et}_{\vec{M}}$ be the sequence of $e^{op}_\msg$ events in order of $\vec{M}$.
Let $\vec{\Et}_{\vec{M},j}$ be the prefix of the first $j$ events in $\vec{M}$.
Let $\Et_{post, -j}$ be the sequence of of events in $\Et_{post}$ without the events of $\vec{\Et}_{\vec{M},j}$. Then define 
\[ \Zt^i_{op,post} = (\fstate^{op}, \vec{\Et}_{op} \cc \vec{\Et}_{\vec{M}, |\vec{M}|} \cc \vec{\Et}_{post, -|\vec{M}|}).\]
The following claim will show that $\Zt^i_{op,post}$ behaves as $\Yt^i_{op,post}$.

\begin{claim}\label{clm:swap-reception}
    Let $\Wt_j = (\fstate^{op}, \vec{ \Et}_{op} \cc \vec{\Et}_{\vec{M},j} \cc \vec{\Et}_{post, -j})$,
    and note that $\Wt_0 = \Yt^i_{op,post}$ and $\Wt_{|\vec M|} = \Zt^i_{op,post}$.
    For each $j~\in~[0, |\vec M|]$, $\Wt_j$ is a well-formed execution with the same final state as $\Yt^i_{op,post}$.
\end{claim}
\begin{proof}
The base case of $j = 0$ follows, since $\Wt_0 = \Yt^i_{op,post}$.
The inductive hypothesis is that the claim is true for some $j \in [0, |\vec{M}|)$.

For the induction step, let $\msg$ be the $(j+1)$st message of $\vec{M}$.
The event just preceding $e^{op}_{\msg}$ in $\Wt_j$ is the reception $e^r_{\msg}$ of $\msg$ by its receiver process $\proc$.
Since applying $G.\msg$ to $\msg$ does not affect $\tau_\msg$, and the reception event acts as identity $\quantum_\msg$, the net effect of the reception and $e^{op}_{\msg}$ is equivalent to the application of $G.\chan$ on $\msg$ occurring before the reception, while the message is still in flight.
Let $\Wt_j'$ be the resultant execution obtained by swapping the events $e^{op}_{\msg}$ and $e^r_{\msg}$, so that the event $e^{op}_{\msg}$ occurs on the message while it is in flight, immediately before it is received at $e^r_{\msg}$. 
The execution $\Wt_j'$, thus, has the same final state as $\Wt_j$. Note that it is \emph{not} true that $\Wt_j'$ is equicausal to $\Wt_j$: this is because we are inverting causally related events $e^{r}_{\msg} \cprec_{\Wt_j} e^{op}_{\msg}$. Nevertheless, the fact that $\Wt_j'$ and $\Wt_j$ have the same final state indicates that this causal relation is somewhat ``illusory.'' 

Now, we take $\Wt_j'$ and further swap $e^{op}_{\msg}$ backwards until it is adjacent to $\vec{\Et}_{M,j}$. Each such swap can be done equicausally by using Lemma~\ref{lem:inversion-swap}, since $e^{op}_{\msg}$ acts only the message in flight and thus has no causal relationship to the event it is being swapped with. The final execution reached after performing these swaps is precisely $\Wt_{j+1}$. This shows that $\Wt_{j+1} \ceq \Wt_j'$, and thus, by Theorem~\ref{thm:equiv}, $\Wt_{j+1}$ has the same final state as $\Wt_j'$ and therefore $\Wt_j$.
\end{proof}

We remark that by construction $\Hist(\Zt^i_{op,post}) = \Hist(\Yt^i_{op,post})$.
Thus, $\Zt^i_{main} = \Yt^i_{pre} \cc \Zt^i_{op,post}$ is a well-formed execution with the same final state as $\Yt^i_{main}$.

Inductively replacing each $\Yt^i_{main}$ with the corresponding $\Zt^i_{main}$, we get
\[\Zt = (\fstatet^0, \Xt^0 \cc \Zt^1_{main} \cc \Xt^1 \cc \Zt^2_{main} \cc \cdots \cc \Zt^k_{main} \cc \Xt^k), \]
and note that $\Hist(\Zt) = \Hist(\Yt)$.

\paragraph{Constructing specification fragment $\Yhat$:}

We now construct a specification execution fragment $\Yhat$ with $\Hist(\Yhat) \simeq \Hist(\Zt)$.
The key idea is to exploit the fact that we can write each $\Zt^i_{op,post}$ as a concatenation 
\[ \Zt^i_{op,post} = \Zt^i_{op,M} \cc \Zt^i_{post},\]
where $\Zt^i_{op,M}$ contains exactly the events in $\Et^i_{op} \cup \Et^i_{\vec{M}}$ and $\Zt^i_{post}$ contains exactly the events in $\Et^i_{post, -|\vec{M}|}$. We use this to construct $\Yhat$ by (1) replacing all the events in each $\Zt^i_{op,M}$ block by a single event that atomically applies the global operation, and (2) omitting all other events that are not contained in our filter set $F$ (i.e. events associated with the ``paraphernalia'' of the algorithm, such as receiving marker messages). 

We now spell this out in slightly more detail. We want to construct a $\Yhat$ that is an execution for $\deltahat$ while having the same history as $\Zt$. 
Recall that the history is simply the temporal sequence of events that lie in the filter set $F$: in our case, events apart from the $G$ operations and associated paraphernalia in the algorithm (e.g. the marker messages). 
In order to preserve this history, we will construct $\Yhat$ by going through the events in $\Zt$ in temporal order and producing corresponding events in $\Yhat$. 
Each event that passes the filter $F$ corresponds to either a transition from the underlying algorithm $\A$, or an invocation or response event. 
All of these have straightforward corresponding transition in $\deltahat$. 
So for every such event, we simply append the corresponding event to $\Yhat$.

The case of events that do not pass the filter $F$ is a bit more complicated, and it is here that we exploit the structure of $\Zt$. Recall that in $\Zt$, for each invocation $i$, all $G$ operations on all processors and channels are grouped together into a block $\Zt^i_{op,M}$. 
We will replace this \emph{entire} block with the corresponding transition $\deltahat$ implementing the global operation $G$ atomically. Apart from these events, the only other events in $\Zt$ that do not pass the filter $F$ are the events where a processor records a message received along a channel before the marker message, and the events where the processor receives a marker message. We simply drop these events from the definition of $\Yhat$.

It is clear to see that this construction produces a $\Yhat$ whose history matches that of $\Zt$---this follows from the linear order in which we process events. To show that $\Yhat$ is a valid execution for $\deltahat$, we argue that each transition in $\Yhat$ is valid. This follows because the atomic global transition in $\deltahat$ accurately simulates the events in $\Zt^i_{op,M}$, while the events that we omit have no effect on the state outside of $\xcstatet$.

\end{proof}

\section{Final Remarks}
\label{sec:conclusion}

The basic take-away message we see from the proof of Theorem~\ref{thm:main} is that classical causality-based arguments are effective for analyzing quantum distributed systems. However, a reader familiar with Chandy and Lamport's original analysis might ask why our proof is longer and more involved than theirs. Interestingly, the answer to this is because of an essentially classical phenomenon which the quantum setting surfaces to our attention. Essentially, the main additional complication in our argument that is not present in~\cite{ChandyLamport} is the reordering of the global operations on messages, which is performed in Claim~\ref{clm:swap-reception}. This step is actually implicit in~\cite{ChandyLamport}, but passes unnoticed because the classical recording operation does not disturb the state, so it can be freely imagined to have occurred at any point in the message's history. In contrast, a general classical or quantum global operation will disturb the state of the message, so its ordering relative to other events is consequential. Moreover, the reason why special reasoning had to be performed for message operations is because of the ``syntactic'' nature of the definition of computational causality, which leads to an apparent causal relation between the global operation on the message in flight and the reception of the message. As remarked in the proof of Claim~\ref{clm:swap-reception}, this causal relation is illusory because the reception event does not touch the state of the message. We conjecture that a more ``semantic'' notion of causality could automatically recognize this and simplify this argument.

Another point that deserves further study is the interpretation of our theorem in terms of probabilities, an issue that once again is relevant classically but unavoidable in the quantum setting. In our theorem, the executions $\Xt$ and $\Yt$ are equicausal, which by Theorem~\ref{thm:equiv} means that $\Xt$ and $\Yt$ have the same final state. Because of our use of subnormalized density matrices, this means that in both $\Xt$ and $\Yt$, the probability that all measurements performed along the way produced the outcomes appearing in the execution must be equal. It is tempting to conclude from this that the distribution over responses produced by running $\At$ matches the distribution produced by running the specification $\deltahat$ in some sense. However, this does not follow: suppose we fix an algorithm for the adversarial scheduler and then run $\At$. Each possible sequence of measurement outcomes $\vec{r}$ from this experiment corresponds to a different execution $\Xt_{\vec{r}}$, and all our theorem says is that each one of these executions can be mapped to an equicausal execution $\Yt_{\vec{r}}$, and an associated $\deltahat$-execution $\Yhat_{\vec{r}}$. However, it is not clear that all these executions correspond to a meaningful experiment in which $\deltahat$ is run in interaction with a fixed adversary. We leave this question open for future work: perhaps ideas from the literature on linearizability and randomized algorithms~\cite{StrongLinearizability} will be relevant to solving it.

\bibliography{quantum-snapshots}

@article{ren2017ground,
  title={Ground-to-satellite quantum teleportation},
  author={Ren, Ji-Gang and Xu, Ping and Yong, Hai-Lin and Zhang, Liang and Liao, Sheng-Kai and Yin, Juan and Liu, Wei-Yue and Cai, Wen-Qi and Yang, Meng and Li, Li and others},
  journal={Nature},
  volume={549},
  number={7670},
  pages={70--73},
  year={2017},
  publisher={Nature Publishing Group UK London}
}

@article{eberhard1978bell,
  title={Bell’s theorem and the different concepts of locality},
  author={Eberhard, Philip Herbert},
  journal={Il Nuovo Cimento B (1971-1996)},
  volume={46},
  number={2},
  pages={392--419},
  year={1978},
  publisher={Springer}
}

@article{eberhard1989quantum,
  title={Quantum field theory cannot provide faster-than-light communication},
  author={Eberhard, Phillippe H and Ross, Ronald R},
  journal={Foundations of Physics Letters},
  volume={2},
  number={2},
  pages={127--149},
  year={1989},
  publisher={Springer}
}

@phdthesis{YunchaoThesis,
title={Shallow Quantum Circuits: Algorithms, Complexity, and
Fault Tolerance},
author={Yunchao Liu},
school={UC Berkeley},
year={2024},
howpublished={\url{https://www2.eecs.berkeley.edu/Pubs/TechRpts/2024/EECS-2024-181.pdf}}
}

@misc{LamportWebsite,
  title = {My Writings
},
    author={Leslie Lamport},
  howpublished = {\url{https://lamport.azurewebsites.net/pubs/pubs.html#chandy}},
  note = {Accessed: 2026-02-16}
}

@article{bravyi2018quantum,
  title={Quantum advantage with shallow circuits},
  author={Bravyi, Sergey and Gosset, David and K{\"o}nig, Robert},
  journal={Science},
  volume={362},
  number={6412},
  pages={308--311},
  year={2018},
  publisher={American Association for the Advancement of Science}
}

@article{wehner2018quantum,
  title={Quantum internet: A vision for the road ahead},
  author={Wehner, Stephanie and Elkouss, David and Hanson, Ronald},
  journal={Science},
  volume={362},
  number={6412},
  pages={eaam9288},
  year={2018},
  publisher={American Association for the Advancement of Science}
}

@article{khatri2021spooky,
  title={Spooky action at a global distance: analysis of space-based entanglement distribution for the quantum internet},
  author={Khatri, Sumeet and Brady, Anthony J and Desporte, Ren{\'e}e A and Bart, Manon P and Dowling, Jonathan P},
  journal={npj Quantum Information},
  volume={7},
  number={1},
  pages={4},
  year={2021},
  publisher={Nature Publishing Group UK London}
}

@book{kitaev2002classical,
  title={Classical and quantum computation},
  author={Kitaev, Alexei Yu and Shen, Alexander and Vyalyi, Mikhail N},
  number={47},
  year={2002},
  publisher={American Mathematical Soc.}
}

@book{nielsen2010quantum,
  title={Quantum computation and quantum information},
  author={Nielsen, Michael A and Chuang, Isaac L},
  year={2010},
  publisher={Cambridge University Press}
}

@article{bell1964einstein,
  title={On the {Einstein} {Podolsky} {Rosen} paradox},
  author={Bell, John S},
  journal={Physics Physique Fizika},
  volume={1},
  number={3},
  pages={195},
  year={1964},
  publisher={APS}
}

@article{brukner2014quantum,
  title={Quantum causality},
  author={Brukner, {\v{C}}aslav},
  journal={Nature Physics},
  volume={10},
  number={4},
  pages={259--263},
  year={2014},
  publisher={Nature Publishing Group UK London}
}

@article{einstein1935can,
  title={Can quantum-mechanical description of physical reality be considered complete?},
  author={Einstein, Albert and Podolsky, Boris and Rosen, Nathan},
  journal={Physical Review},
  volume={47},
  number={10},
  pages={777},
  year={1935},
  publisher={APS}
}

@article{LamportTimeClocksOrderingOfEvents,
  author       = {Leslie Lamport},
  title        = {Time, Clocks, and the Ordering of Events in a Distributed System},
  journal      = {Commun. {ACM}},
  volume       = {21},
  number       = {7},
  pages        = {558--565},
  year         = {1978},
  url          = {https://doi.org/10.1145/359545.359563},
  doi          = {10.1145/359545.359563},
  bibsource    = {dblp computer science bibliography, https://dblp.org}
}

@article{ChandyLamport,
  author       = {K. Mani Chandy and
                  Leslie Lamport},
  title        = {Distributed Snapshots: Determining Global States of Distributed Systems},
  journal      = {{ACM} Trans. Comput. Syst.},
  volume       = {3},
  number       = {1},
  pages        = {63--75},
  year         = {1985},
  url          = {https://doi.org/10.1145/214451.214456},
  doi          = {10.1145/214451.214456},
  bibsource    = {dblp computer science bibliography, https://dblp.org}
}

@inproceedings{Jayanti2025,
author = {Jayanti, Siddhartha Visveswara},
title = {On Interplanetary and Relativistic Distributed Computing},
year = {2025},
isbn = {9798400718854},
publisher = {Association for Computing Machinery},
address = {New York, NY, USA},
url = {https://doi.org/10.1145/3732772.3733563},
doi = {10.1145/3732772.3733563},
booktitle = {Proceedings of the ACM Symposium on Principles of Distributed Computing},
pages = {192–202},
numpages = {11},
location = {Hotel Las Brisas Huatulco, Huatulco, Mexico},
series = {PODC '25}
}

@inproceedings{GilbertGolabRelativisticDistributedSystems,
  author       = {Seth Gilbert and
                  Wojciech M. Golab},
  editor       = {Fabian Kuhn},
  title        = {Making Sense of Relativistic Distributed Systems},
  booktitle    = {Distributed Computing - 28th International Symposium, {DISC} 2014,
                  Austin, TX, USA, October 12-15, 2014. Proceedings},
  series       = {Lecture Notes in Computer Science},
  volume       = {8784},
  pages        = {361--375},
  publisher    = {Springer},
  year         = {2014},
  url          = {https://doi.org/10.1007/978-3-662-45174-8\_25},
  doi          = {10.1007/978-3-662-45174-8\_25},
  bibsource    = {dblp computer science bibliography, https://dblp.org}
}

@inproceedings{StrongLinearizability,
  author    = {Wojciech M. Golab and
               Lisa Higham and
               Philipp Woelfel},
  editor    = {Lance Fortnow and
               Salil P. Vadhan},
  title     = {Linearizable implementations do not suffice for randomized distributed
               computation},
  booktitle = {Proceedings of the 43rd {ACM} Symposium on Theory of Computing, {STOC}
               2011, San Jose, CA, USA, 6-8 June 2011},
  pages     = {373--382},
  publisher = {{ACM}},
  year      = {2011},
  url       = {https://doi.org/10.1145/1993636.1993687},
  doi       = {10.1145/1993636.1993687},
  bibsource = {dblp computer science bibliography, https://dblp.org}
}

@inproceedings{A97,
 author = {Abadi, Mart\'{\i}n and Gordon, Andrew D.},
 title = {A Calculus for Cryptographic Protocols: The Spi Calculus},
 booktitle = {Proceedings of the 4th ACM Conference on Computer and Communications Security},
 series = {CCS '97},
 year = {1997},
 isbn = {0-89791-912-2},
 location = {Zurich, Switzerland},
 pages = {36--47},
 numpages = {12},
 url = {http://doi.acm.org/10.1145/266420.266432},
 doi = {10.1145/266420.266432},
 acmid = {266432},
 publisher = {ACM},
 address = {New York, NY, USA},
}

@article{H,
 author = {Herlihy, Maurice},
 title = {Wait-free Synchronization},
 journal = {ACM Trans. Program. Lang. Syst.},
 issue_date = {Jan. 1991},
 volume = {13},
 number = {1},
 month = {January},
 year = {1991},
 issn = {0164-0925},
 pages = {124--149},
 numpages = {26},
 url = {http://doi.acm.org/10.1145/114005.102808},
 doi = {10.1145/114005.102808},
 acmid = {102808},
 publisher = {ACM},
 address = {New York, NY, USA},
}

@inproceedings{G,
 author = {Gibbons, P. B.},
 title = {A More Practical PRAM Model},
 booktitle = {Proceedings of the First Annual ACM Symposium on Parallel Algorithms and Architectures},
 series = {SPAA '89},
 year = {1989},
 isbn = {0-89791-323-X},
 location = {Santa Fe, New Mexico, USA},
 pages = {158--168},
 numpages = {11},
 url = {http://doi.acm.org/10.1145/72935.72953},
 doi = {10.1145/72935.72953},
 acmid = {72953},
 publisher = {ACM},
 address = {New York, NY, USA},
}

@inproceedings{dufoulon-magniez-pandurangan-podc-2025,
author = {Dufoulon, Fabien and Magniez, Fr\'{e}d\'{e}ric and Pandurangan, Gopal},
title = {Quantum Communication Advantage for Leader Election and Agreement},
year = {2025},
isbn = {9798400718854},
publisher = {Association for Computing Machinery},
address = {New York, NY, USA},
url = {https://doi.org/10.1145/3732772.3733509},
doi = {10.1145/3732772.3733509},
booktitle = {Proceedings of the ACM Symposium on Principles of Distributed Computing},
pages = {230–240},
numpages = {11},
location = {Hotel Las Brisas Huatulco, Huatulco, Mexico},
series = {PODC '25}
}

@article{denchev-pandurangan-2008,
author = {Denchev, Vasil S. and Pandurangan, Gopal},
title = {Distributed quantum computing: a new frontier in distributed systems or science fiction?},
year = {2008},
issue_date = {September 2008},
publisher = {Association for Computing Machinery},
address = {New York, NY, USA},
volume = {39},
number = {3},
issn = {0163-5700},
url = {https://doi.org/10.1145/1412700.1412718},
doi = {10.1145/1412700.1412718},
journal = {SIGACT News},
month = sep,
pages = {77–95},
numpages = {19}
}

@inproceedings{balliu-et-al-2025,
author = {Balliu, Alkida and Brandt, Sebastian and Coiteux-Roy, Xavier and d'Amore, Francesco and Equi, Massimo and Le Gall, Fran\c{c}ois and Lievonen, Henrik and Modanese, Augusto and Olivetti, Dennis and Renou, Marc-Olivier and Suomela, Jukka and Tendick, Lucas and Veeren, Isadora},
title = {Distributed Quantum Advantage for Local Problems},
year = {2025},
isbn = {9798400715105},
publisher = {Association for Computing Machinery},
address = {New York, NY, USA},
url = {https://doi.org/10.1145/3717823.3718233},
doi = {10.1145/3717823.3718233},
booktitle = {Proceedings of the 57th Annual ACM Symposium on Theory of Computing},
pages = {451–462},
numpages = {12},
location = {Prague, Czechia},
series = {STOC '25}
}

\end{document}